\documentclass[aps, prx, reprint, amsmath,amssymb,showpacs,floatfix,longbibliography, twocolumn, superscriptaddress,nofootinbib]{revtex4-1}
\usepackage{cmap} 
\usepackage{CJKutf8}
\usepackage[utf8]{inputenc}
\usepackage{times} 
\usepackage{graphicx}
\usepackage{dcolumn}
\usepackage{bm}
\usepackage{color}
\usepackage{xcolor}
\usepackage{hyperref}
\usepackage{enumitem}
\usepackage{cancel}
\usepackage{braket}
\usepackage{stmaryrd}
\usepackage{tabularx}
\DeclareMathOperator*{\argmin}{arg\,min}
\usepackage{amsthm}
\usepackage{comment}
\hypersetup{colorlinks=true,citecolor=blue,linkcolor=blue, urlcolor=blue}
\hypersetup{linktocpage}
\newcolumntype{M}[1]{>{\centering\arraybackslash}m{#1}}
\newcolumntype{N}{@{}m{0pt}@{}}
\usepackage{environ}
\usepackage{MnSymbol}

\bibpunct{[}{]}{,}{n}{}{}

\usepackage{amsfonts,amssymb,amsmath}
\usepackage[T1]{fontenc}
\usepackage{tikz}

\usepackage{booktabs} 

\let\originalleft\left
\let\originalright\right
\renewcommand{\left}{\mathopen{}\mathclose\bgroup\originalleft}
\renewcommand{\right}{\aftergroup\egroup\originalright}

\NewEnviron{eqs}{%
\begin{equation}\begin{split}
    \BODY
\end{split}\end{equation}
}
\newtheorem{lemma}{Lemma}
\newcommand{\enc}{U_{\text{enc}}}
\newcommand{\aenc}{A_{\text{enc}}}

\begin{document}
\begin{CJK*}{UTF8}{gbsn}

\title{Qubit-oscillator concatenated codes: decoding formalism \& code comparison}

\author{Yijia Xu}
\email[E-mail: ]{yijia@umd.edu}

\affiliation {Joint Center for Quantum Information and Computer Science, NIST and University of Maryland, College Park, Maryland 20742, USA}
\affiliation {Joint Quantum Institute, NIST and University of Maryland, College Park, Maryland 20742, USA}
\affiliation {Department of Physics, University of Maryland, College Park, Maryland 20742, USA}
\affiliation {Institute for Physical Science and Technology, University of Maryland, College Park, Maryland 20742, USA}

\author{Yixu Wang}
\email[E-mail: ]{wangyixu@terpmail.umd.edu}

\affiliation {Department of Physics, University of Maryland, College Park, Maryland 20742, USA}
\affiliation{Maryland Center for Fundamental Physics, University of Maryland,
College Park, MD 20742, USA}

\author{En-Jui Kuo}

\affiliation {Joint Center for Quantum Information and Computer Science, NIST and University of Maryland, College Park, Maryland 20742, USA}
\affiliation {Joint Quantum Institute, NIST and University of Maryland, College Park, Maryland 20742, USA}
\affiliation {Department of Physics, University of Maryland, College Park, Maryland 20742, USA}

\author{Victor V. Albert}
\affiliation {Joint Center for Quantum Information and Computer Science, NIST and University of Maryland, College Park, Maryland 20742, USA}
\affiliation {Department of Physics, University of Maryland, College Park, Maryland 20742, USA}

\date{\today}

\begin{abstract}

Concatenating bosonic error-correcting codes with qubit codes can substantially boost the error-correcting power of the original qubit codes. 
It is not clear how to concatenate optimally, given there are several bosonic codes and concatenation schemes to choose from, including the recently discovered GKP-stabilizer codes [\textcolor{black}{Phys. Rev. Lett. 125, 080503 (2020)}] that allow protection of a logical bosonic mode from fluctuations of the mode's conjugate variables. 
We develop efficient maximum-likelihood decoders for and analyze the performance of three different concatenations of codes taken from the following set: qubit stabilizer codes, analog/Gaussian stabilizer codes, GKP codes, and GKP-stabilizer codes.
We benchmark decoder performance against additive Gaussian white noise, corroborating our numerics with analytical calculations.
We observe that the concatenation involving GKP-stabilizer codes outperforms the more conventional concatenation of a qubit stabilizer code with a GKP code in some cases.
We also propose a GKP-stabilizer code that suppresses fluctuations in both conjugate variables without extra quadrature squeezing, and formulate qudit versions of GKP-stabilizer codes.
\end{abstract}
\maketitle
\end{CJK*}

\begin{figure*}[t]
    \raggedright
    \includegraphics[width=\textwidth]{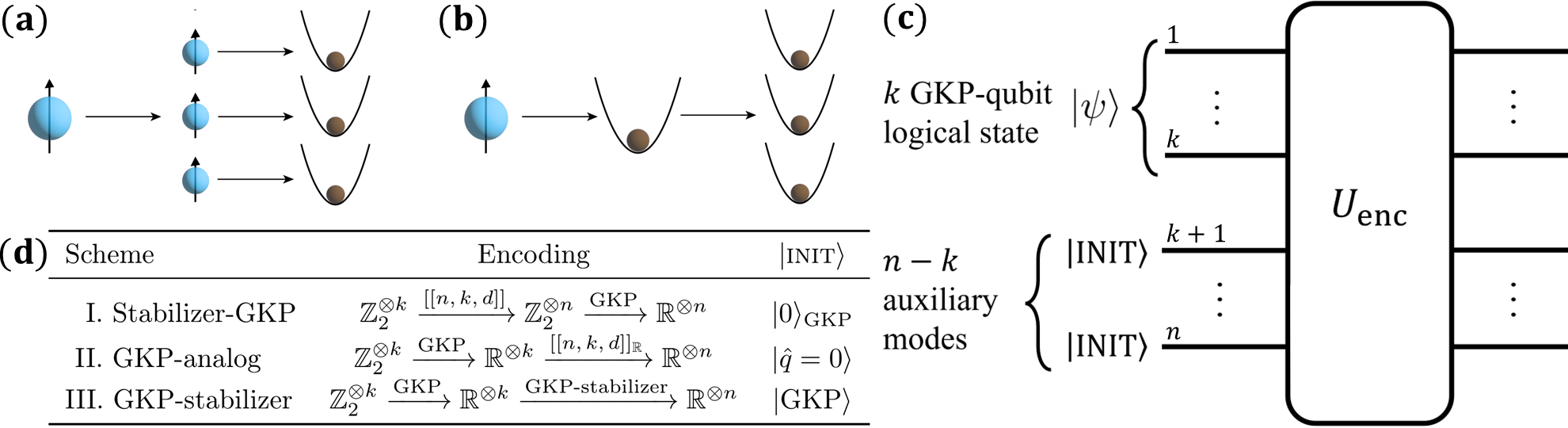}
    \caption{
    \textsc{Encoder summary:}
    {\bf (a)} Sketch of the conventional discrete-variable (DV) and continuous-variable (CV) concatenated encoding class ``DV-DV-CV'', where logical qubits are encoded into an outer multi-qubit code, and each qubit of the outer code is further encoded into a single physical mode. 
    {\bf (b)} Alternative concatenated encoding class ``DV-CV-CV'', where each logical qubit is encoded into an outer single-mode code, and the modes are further encoded into a multi-mode CV code. 
    {\bf (c)} Encoding maps for our concatenated codes can be formulated as Gaussian operations \(\enc\) acting on \(k\) ``logical'' modes encoded in GKP states and \(n-k\) modes in a fixed initial state \(|\textsc{init}\rangle\).
    {\bf (d)} Table of the concatenated encodings considered in this work.
    The ``Stabilizer-GKP'' DV-DV-CV encoding scheme I is the conventional concatenation of a qubit \(\llbracket n,k,d\rrbracket \) stabilizer outer code \cite{gottesman1997stabilizer,calderbank97} with a GKP inner code \cite{gkp}. 
    The ``GKP-analog'' DV-CV-CV encoding scheme II is a concatenation of a single-mode GKP outer code with an analog \(\llbracket n,k,d\rrbracket_{\mathbb{R}}\) stabilizer code \cite{braunstein98,lloyd98,Gu_2009,eczoo_analog_stabilizer}. 
    The ``GKP-stabilizer'' DV-CV-CV encoding scheme III is a concatenation of a single-mode GKP outer code with an \(n\)-mode GKP-stabilizer code \cite{Noh_many_oscillator}.
    The initial state \(|\textsc{init}\rangle\) for the encoding circuits for the three codes, shown in the third column of the table, is a GKP logical-zero state \(|0\rangle_{\text{GKP}}\) in Eq.~\eqref{eq:gkp0}, the position state \(\left|\hat{q}=0\right\rangle\), or the canonical GKP state \(|\text{GKP}\rangle\) in Eq.~\eqref{eq:canonicalgkp}, respectively.
    }    \label{fig:encoding-overview}
\end{figure*}
%

\section{Introduction}

Quantum error correction (QEC) is one of the most challenging tasks in building large-scale quantum computers. Its basic idea is to encode a few logical degrees of freedom into a larger physical system. Quantum error correction is required if we are to scale up quantum devices both in terms of the length of a quantum communication link or the computational power of a quantum computer.

On one side of the large field of error correction are the well-established qubit or discrete-variable (DV) stabilizer codes \cite{gottesman1997stabilizer,calderbank97}, some of which allow one to suppress noise to arbitrary accuracy given sufficient physical resources once the physical error rate is below certain threshold value --- a manifestation of the all-important threshold theorem \cite{abanov99ft,Knill98ft,Kitaev03ft,shor96ft}. 

On the other side are the bosonic codes \cite{eczoo_bosonic}, which are instead typically designed to satisfy existing resource constraints and which are naturally compatible with several continuous-variable (CV) quantum platforms, including microwave cavities \cite{leghtas2015confining,ofek2016extending,wang2016schrodinger,touzard_coherent,hu2019quantum,campagne2020quantum,grimm2019kerr,lescanne2020exponential,gertler2021protecting,break_even2022,gertler2022experimental,ni2023beating} and motional degrees of freedom of trapped ions \cite{fluhmann18sequential,fluhmann2019encoding,fluhmann20direct,de2022error}. The analog information given by the infinite-dimensional nature of the bosonic Hilbert space also allows for error-correction schemes not available in the DV world \cite{fukui17,fukui18,fukui18surface,terhal19,fukui2019high,pattison2021improved,albert_summary}.

It is fruitful to consider the marriage of the abstract yet scalable qubit paradigm with the practically oriented bosonic paradigm, in hopes of bringing out \textcolor{black}{the} advantages of both.
This direction has so far proven to be promising, with the analog syndrome information provided by a CV layer of correction substantially increasing the correcting power of the outer DV layer in a concatenated scheme.
For example, there have been corroborating studies on concatenating a particular bosonic code --- the GKP code \cite{gkp,terhal2020towards,Puri_gkp} --- with several DV codes, such as the repetition code \cite{fukui17}, $\llbracket 4,2,2\rrbracket$ code \cite{fukui17,fukui18}, surface code \cite{terhal19,terhal2020towards,noh2020fault,noh22surface,fukui18surface,wang19,Hanggli_biased,pattison2021improved}, color code \cite{color_gkp}, XZZX surface code \cite{xzzx_gkp}, and lifted-product QLDPC code \cite{qldpc_gkp}.

Given the abundance of bosonic codes \cite{munro_cat,albert_binomial,grimsmo_rotation,albert2019pair}, there remain other less well-studied ways of encoding qubits into modes that may outperform the aforementioned established DV-CV concatenation scheme in either scalability or resource efficiency. 
Moreover, thanks to the discovery of GKP-stabilizer codes \cite{Noh_many_oscillator}, it is possible to suppress small fluctuations of a logical mode's position and momentum quadratures by encoding it into several physical modes. 
Previous studies show that analog-stabilizer codes utilizing only Gaussian resources are limited, \textcolor{black}{and} cannot protect against Gaussian noise \cite{terhal19,niset09,eisert02}. 
The key concept of GKP-stabilizer codes is introducing auxiliary modes initialized in GKP states \cite{gkp,Terhal_grid,Hanggli_threshold} as non-Gaussian resources to circumvent these no-go theorems. 
Hence, it is interesting to investigate the performance of GKP-stabilizer codes when they are used for protecting a discrete-variable (GKP) subspace of a CV mode.  
Given the recent advances in the realization of GKP codes \cite{campagne2020quantum, break_even2022,fluhmann18sequential, fluhmann2019encoding, de2022error}, and bosonic gate operations \cite{Furusawa_Sum2008,chen2023scalable}, theoretical and numerical studies of GKP codes and their various concatenations is an imperative topic.
A goal of this work is to begin to probe whether utilizing this new code in a concatenation scheme can provide advantages over established schemes as well as schemes using other mode-into-mode bosonic codes \cite{lloyd98,braunstein98,braunstein1998quantum,aoki2009quantum,hayden2016spacetime,hayden21error,faist20continuous,woods2020continuous}. 

The performance of single-mode bosonic codes has been studied in Ref.~\cite{Albert_performance}. However, analytical and numerical studies of multi-mode bosonic codes are far from being exhausted due to the lack of a unified formalism for CV codes. In particular, the motivational GKP-stabilizer work \cite{Noh_many_oscillator} focuses substantially on proof-or-principle examples and lacks a general decoder. In this work, we also provide a unified framework to describe maximum-likelihood decoding against independent and identically distributed (i.i.d.) Gaussian quadrature noise for three different concatenation schemes, one of which includes GKP-stabilizer codes.

\section*{Summary of results}

We study three types of concatenated encodings of qubits into modes that consist of various combinations of \textcolor{black}{the} qubit and bosonic \cite{eczoo_bosonic_stabilizer} stabilizer codes (see Fig. \ref{fig:encoding-overview}). Our motivation is to shed light on which combinations of such qubit-into-mode and mode-into-mode encodings perform the best under standard noise models.

Our first encoding, which we call scheme I, consists of first encoding qubits into a qubit \(\llbracket n,k,d\rrbracket \) stabilizer code \cite{gottesman1997stabilizer,calderbank97} and then further encoding each qubit into its own mode using the GKP qubit-into-mode code \cite{gkp}. This scheme is the one most commonly used when concatenating qubit and bosonic codes \cite{fukui17,fukui18,terhal19,terhal2020towards,noh2020fault,noh22surface,fukui18surface,wang19,Hanggli_biased,color_gkp,xzzx_gkp,qldpc_gkp}. 
The second encoding --- scheme II --- is essentially the reverse of the first: each qubit is first encoded into a ``logical mode'' using the GKP code, which is subsequently encoded into an analog \(\llbracket n,k,d\rrbracket _{\mathbb{R}}\) bosonic stabilizer code \cite{braunstein98,lloyd98,Gu_2009,eczoo_analog_stabilizer}. The third encoding --- scheme III --- substitutes the analog code of scheme II with a GKP-stabilizer mode-into-mode code \cite{Noh_many_oscillator}.

We first observe that encoding maps for all three schemes are of similar type (see Sec. \ref{subsec:encoding}).
We show that encoding maps for all three schemes can be formulated as a Gaussian unitary acting on \(k\) logical GKP qubits tensored with \(n-k\) auxiliary resource states that are either GKP states or position eigenstates (see Fig. \ref{fig:encoding-overview}).

As for decoding, we recast the problem of finding the most likely error under zero-mean Gaussian displacement noise as a linear optimization problem for scheme II, and simplify said problem for scheme III to a closely related linear optimization (see Sec. \ref{sec:EC_dec}). 
Both maximum-likelihood optimizations can be solved exactly in a time that is polynomial in the total number of modes \(n\) of the encoding, yielding an efficient decoder for analog- and GKP-stabilizer codes in the process. 

To benchmark the three schemes, we numerically compare each scheme using repetition, $\llbracket5,1,3\rrbracket$, $\llbracket7,1,3\rrbracket$ (Steane) \cite{steane_code} and $\llbracket9,1,3\rrbracket$ (Shor) \cite{shor_code} codes (see Sec. \ref{sec:example}). 
To generate the examples, we fix the GKP codes to be the same for each scheme, meaning that we are left with the choice of the qubit stabilizer code for scheme I, analog stabilizer code for scheme II, and the GKP-stabilizer code for scheme III. We compare the performance of different schemes with respect to a fixed stabilizer code. 
For example, the comparison based on the repetition code (see Sec. \ref{sec:repetition}) utilizes the qubit repetition code for scheme I, its analog version \cite{lloyd98} for scheme II, and the GKP-repetition code \cite{Noh_many_oscillator} for scheme III.
The comparison based on the five-qubit code (see Sec. \ref{sec:513}) utilizes, respectively, the five-qubit code \cite{laflamme_513}, its analog version \cite{braunstein98}, and the GKP-\(\llbracket 5,1,3\rrbracket \) mode-into-mode code \cite{Noh_many_oscillator}. 
We corroborate some of our numerics by analytically calculating logical error probabilities (see Sec. \ref{sec:calc_logical_rate}).

Scheme II is used, in part, for reference in our numerical comparisons because the outer analog codes are ineffective against Gaussian noise \cite{terhal19,niset09,eisert02}.
However, since the inner code \textcolor{black}{uses} non-Gaussian resources, these no-go theorems technically do not apply to the entire scheme.

We observe that scheme III outperforms scheme I in the repetition-code comparison by a constant factor in the intermediate noise regime despite the number of syndromes and decoding complexity being substantially higher.
However, when the number of physical modes grows, as in the following comparisons of the $\llbracket5,1,3\rrbracket, \llbracket7,1,3\rrbracket$ and $\llbracket9,1,3\rrbracket$ codes,  scheme I easily surpasses other schemes. 

Scheme II performs the worst in the repetition and $\llbracket5,1,3\rrbracket$ code comparisons, while attaining second place for $\llbracket7,1,3\rrbracket$ and $\llbracket9,1,3\rrbracket$ comparisons. 
Despite the existing no-go theorems \cite{terhal19,niset09,eisert02} showing the inability of analog stabilizer codes to correct the i.i.d Gaussian noise considered in this work, scheme II still exhibits error correction ability. This is likely due to the fact that no-go theorem \textcolor{black}{applies} only to a part of scheme II and that the entire scheme \textcolor{black}{uses} non-Gaussian resources in the form of a logical GKP encoding. 
The performance of scheme II is affected by the interplay between the deformation of the logical noise quadrature by the analog stabilizer encoding and the lattice shape of the initial logical GKP states.

Our findings for the repetition-code comparison highlight that scheme III can be on par with the more widely used scheme I while being hardware-efficient in that it requires fewer syndromes, especially when the physical mode number is small.
When the system size grows, scheme III was surpassed by scheme I in $\llbracket5,1,3\rrbracket$ code and further by scheme II in $\llbracket7,1,3\rrbracket$ and $\llbracket9,1,3\rrbracket$ codes.  We postulate that this phenomenon happens because scheme-III encoding suffers from ``error concentration'' whenever a GKP-stabilizer code with high stabilizer weights is used. Since scheme III does not have an outer layer of GKP error correction like scheme I, directly measuring a high-weight stabilizer can cause noise on each mode in the stabilizer's support to add up to an uncorrectable error.

Our encoding and decoding schemes can handle general GKP code lattices. 
This is demonstrated by a study of the error rate with respect to GKP lattice shapes, taking the repetition code as an example (see Appendix \ref{appendix:more_numerics}). \textcolor{black}{From this perspective}, the three schemes are sample points from a family of continuously deformed schemes. 
The ability of deformation enables us to design (outer) encoding methods \textcolor{black}{adapted} to various (inner) stabilizer codes and error models.

Finally, we propose a variant of the GKP-repetition code, which simultaneously suppresses position and momentum error without extra quadrature squeezing (see Sec. \ref{sec:unbiased_code}). Previous work \cite{Noh_many_oscillator} either requires squeezing to suppress both quadratures, or achieves similar suppression in only one quadrature via the GKP-repetition code. We also present a generalization of GKP-stabilizer codes to qudits in Appendix.~\ref{appendix:qudit_version}.

\section{Encodings \& error model}
\label{subsec:encoding}

We describe encoding circuits for the three schemes listed in Fig. \ref{fig:encoding-overview}(d). We finish this section with a description of the displacement error model chosen for our comparison.

\subsection{Scheme I: stabilizer-GKP encoding}
In this encoding, one first makes use of an $\llbracket n,k,d\rrbracket $ stabilizer code \cite{gottesman1997stabilizer,calderbank97} to encode $k$ logical qubits into $n$ physical qubits. Then, each of the $n$ physical qubits is encoded into a harmonic oscillator using the GKP code corresponding to a square lattice in phase space \cite{gkp}. The combined encoding map is from the space of \(k\) qubits into that of \(n\) oscillators,
\begin{eqs}
\text{Stabilizer-GKP:}\quad &\mathbb{Z}_2^{\otimes k} \xrightarrow[]{\text{$\llbracket n,k,d\rrbracket $}} \mathbb{Z}_2^{\otimes n} \xrightarrow[]{\text{GKP}} \mathbb{R}^{\otimes n}~,
\end{eqs}
where \(\mathbb{Z}_2\) (\(\mathbb R\)) represents a qubit (mode).

The stabilizer-GKP encoding can be performed by the following procedure, illustrated in the left panel of Fig.~\ref{fig:encoding-overview} (a). Given $n$ harmonic oscillator modes, we prepare them in GKP states, in which $k$ modes carry logical information of the $k$ logical qubits, and the remaining $(n-k)$ auxiliary modes are in the logical-zero state of the square-lattice GKP code,
\begin{equation}
|0\rangle_{\text{GKP}}=\sum_{n\in\mathbb{Z}}\ket{\hat{q}=2n\sqrt{\pi}}\label{eq:gkp0},
\end{equation}
where \(|\hat{q}=k\rangle\) is the non-normalizable oscillator position state at position \(k\).
Then, we act with a Gaussian circuit \(\enc\) to perform the $\llbracket n,k,d\rrbracket $ stabilizer encoding at the level of GKP qubits. A codeword of the resulting code is simultaneously stabilized by the stabilizers of inner GKP codes and outer embedded qubit stabilizer codes.

\subsection{Scheme II: analog-stabilizer encoding}

In this encoding, each of the $k$ logical qubits is first encoded into a mode using the GKP code. Afterwards, the $k$ ``logical modes'' are encoded into $n$ physical modes using an $\llbracket n,k,d\rrbracket _{\mathbb R}$ analog stabilizer code \cite{braunstein98,lloyd98,Gu_2009,eczoo_analog_stabilizer}. The combined encoding map can be represented as the following,
\begin{eqs}
\text{GKP-analog:}\quad &\mathbb{Z}_2^{\otimes k} \xrightarrow[]{\text{GKP}} \mathbb{R}^{\otimes k} \xrightarrow[]{\text{$\llbracket n,k,d\rrbracket _{\mathbb R}$}} \mathbb{R}^{\otimes n}~.
\end{eqs}

Despite the fact that scheme I and II encodings arise from different concatenation orders [see Fig. \ref{fig:encoding-overview}(a)], the difference in the circuit-level implementation lies only in the initial state of the auxiliary modes.
In other words, the encoding of this scheme can be performed by the same circuit as the previous encoding from Fig. \ref{fig:encoding-overview}(c), but with the \(n-k\) ancillary modes each initialized in the position state $|\hat{q}=0\rangle$.

Before the Gaussian unitary \(\enc\) is applied in the aforementioned circuit, the initial state is stabilized by GKP stabilizers acting on the first $k$ modes and annihilated by position operators $\{\hat{q}_j~|~k < j \leq n\}$ of the auxiliary \(n-k\) modes.
After the unitary is applied, codewords are simultaneously stabilized by logical GKP stabilizers and annihilated by the analog code's \(n-k\) \textit{nullifiers} \(\enc\hat{q}_j \enc^\dagger\). 

Since \(\enc\) is a Gaussian transformation, its action on the mode's position (\(\hat q\)) and momentum (\(\hat p\)) quadrature operators can equivalently be represented as a \(2n\)-dimensional symplectic matrix \(\aenc\) acting on the \(2n\)-dimensional vector of operators \cite{braunstein1998quantum,serafini},
\begin{eqs}\label{eq:symplectic}
    \enc\vec{r}\enc^{\dagger}=\aenc\vec{r}\quad\text{with} \quad \vec{r} = (\hat{q}_1,\cdots,\hat{q}_n,\hat{p}_1,\cdots,\hat{p}_n)^T,
\end{eqs}
where \(\vec{v}^T\) is the transpose of \(\vec{v}\).
Determining how a particular quadrature \(j\) transforms under \(\enc\) amounts to taking the \(j\)th component of both sides. On the right-hand side, this yields an inner product of the \(j\)th row of \(\aenc\) with \(\vec{r}\). 

We will often be interested in how a particular subset of quadratures transforms, for which we only need the set of corresponding rows of \(\aenc\). For such purposes, it is convenient to decompose the encoding matrix into four rectangular submatrices \cite[Appendix. E]{terhal19},
\begin{eqs}
A_{\text{enc}}=\begin{pmatrix}Q\\
G\\
P\\
D
\end{pmatrix}=\begin{pmatrix}k\times2n\text{ matrix}\\
(n-k)\times2n\text{ matrix}\\
k\times2n\text{ matrix}\\
(n-k)\times2n\text{ matrix}
\end{pmatrix}~.
\end{eqs}
The submatrix combinations relevant to schemes II and III are
\begin{eqs}\label{eq:A123}
    A_1=G,~
    A_2=\begin{pmatrix}
    Q\\
    P
    \end{pmatrix},~
    A_{3}=\begin{pmatrix}
    G\\
    D
    \end{pmatrix}~.
\end{eqs}
For the analog stabilizer encoding of scheme II, \(A_1\) represents how the \(n-k\) auxiliary position operators \(\hat{q}_j\) for \(j\in\{k+1,\cdots,n\}\) are transformed into nullifiers \(\enc\hat{q}_j \enc^\dagger\), and \(A_2\) determines how the positions and momenta of \(k\) logical modes are encoded.

\subsection{Scheme III: GKP-stabilizer encoding}

This encoding is a modification of scheme II such that the mode-into-mode outer encoding is now a GKP-stabilizer code \cite{Noh_many_oscillator} (see also Refs. \cite{Terhal_grid,Hanggli_threshold}). The corresponding circuit is yet again of the same type as that depicted in Fig. \ref{fig:encoding-overview}(c), but with the auxiliary modes initialized in the so-called \textit{canonical GKP state} (a.k.a. grid state or trivial GKP code)
\begin{equation}\label{eq:canonicalgkp}
    |\textsc{gkp}\rangle=\sum_{n\in\mathbb{Z}}|\hat{q}=n\sqrt{2\pi}\rangle~.
\end{equation}
This state is the unique simultaneous eigenstate of the canonical-GKP stabilizers $e^{i\sqrt{2\pi}\hat q}$ and $e^{-i\sqrt{2\pi}\hat p}$ with eigenvalue \(+1\), spanning the one-dimensional codespace of the trivial square-lattice GKP code.
The state differs from the logical square-lattice GKP state \eqref{eq:gkp0} in the spacing between the superposed position states.

The canonical GKP state can be transformed from $\ket{0}_{\text{GKP}}$ via squeezing. 
As such, the encoding of Scheme III only differs from that of Scheme I by single-mode squeezing acting on each auxiliary mode. 
Encoded states are stabilized by the \(2(n-k)\) canonical-GKP stabilizers as well as the \(2k\) square-lattice GKP stabilizers of the first \(k\) modes, all conjugated by \(\enc\).

\subsection{Displacement error model}\label{sec: displacement_error}

We adopt a standard error model throughout this paper. 
In order to provide a baseline code comparison, we assume that the encoding, syndrome measurement, and decoding for each scheme are noiseless.
The only source of noise comes after the encoding, when a displacement noise (\textit{a.k.a.} additive Gaussian white-noise \cite{Albert_performance,Noh_thermal_loss,Banaszek2020quantum}) channel is applied on the position and momentum quadratures of each of the \(n\) modes. The position and momentum of each mode \(j\in\{1,2,\cdots,n\}\) are shifted by a random fluctuation \(\xi\) as
\begin{subequations}\label{eq:displacement_noise}
\begin{align}
    {q}_{j}&\rightarrow{q}_{j}+\xi_{j}\\{p}_{j}&\rightarrow{p}_{j}+\xi_{j+n}~.
\end{align}
\end{subequations}
We collect all fluctuations in a \(2n\)-dimensional \textit{noise vector} \begin{eqs}\label{eq:noise_vector}
    \vec{\xi} = (\xi_1,\xi_2,\cdots,\xi_{2n})^T = (\xi_{1,q},\cdots,\xi_{n,q},\xi_{1,p},\cdots,\xi_{n,p})^T.
\end{eqs}
We use either form for the above vector components throughout the manuscript, depending on whether we want to specify if a given quadrature is a position or a momentum.

We develop our decoding formalism with the assumption that the amplitudes of the \(2n\) displacement errors \(\xi_{\ell}\) for \(\ell\in\{1,2,\cdots, 2n\}\) are i.i.d.~Gaussian random variables with the same zero mean and standard deviation \textcolor{black}{$\sigma$}. 

\begin{figure}
    \centering
    \includegraphics[width=0.485\textwidth]{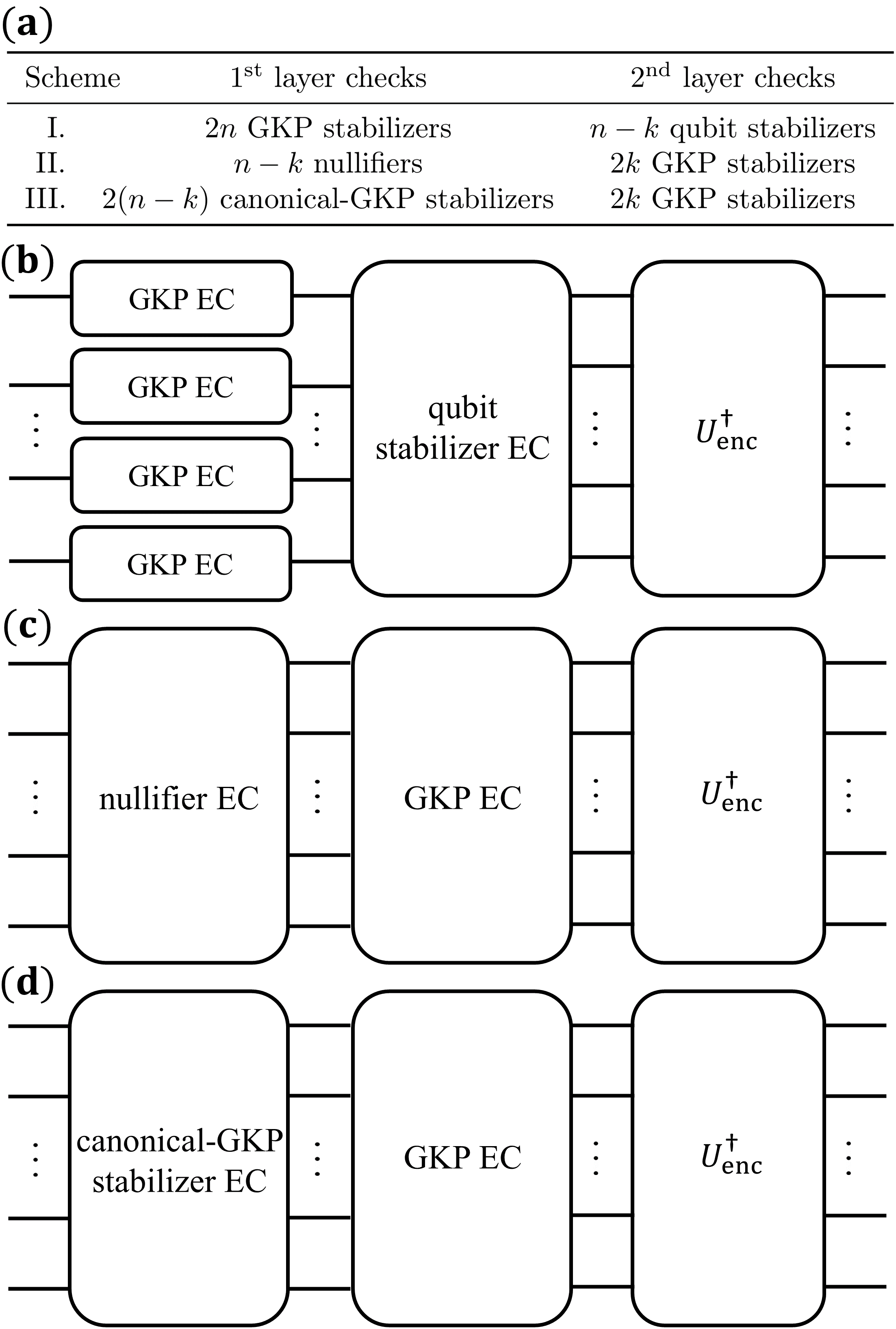}
    \caption{
    \textsc{Decoder summary:} 
    {\textbf{(a)}} Table of check operators that are measured in each of the two layers of error correction of schemes I, II, and III described in Fig. \ref{fig:encoding-overview}. Corresponding circuits for the three schemes are sketched out in \textbf{(b)}, \textbf{(c)}, and \textbf{(d)}, respectively. Each circuit consists of two layers of correction, followed by the inverse encoding map that maps the processed logical information back into the first \(k\) modes depicted in Fig. \ref{fig:encoding-overview}(c).
    }
    \label{fig:decoding}
\end{figure}

\section{Error correction and decoding}
\label{sec:EC_dec}

We summarize the error-correction (EC) processes for the three schemes outlined in Fig. \ref{fig:encoding-overview}. Each round consists of two EC layers --- one for the inner code and one \textcolor{black}{for} the outer. The check operators measured in each layer are listed in Fig. \ref{fig:decoding}(a), and correction circuits for each scheme are in Fig.~\ref{fig:decoding}(b-d), respectively.

For each layer, $\xi$ denotes the actual error that is applied to the system; $z$ denotes a syndrome measurement outcome, which is either real-valued in the case of nullifier-based correction, or a phase in the case of GKP-based correction; and $\xi_*$ denotes the most likely displacement error deduced from the syndrome. 

Our decoding optimizations for schemes II and III are solved by inverting a matrix whose dimensions are at most \(2n\), yielding a polynomial-time \cite{tveit2003complexity} decoder.

\subsection{Scheme I: stabilizer-GKP decoding}\label{subsubsec:scheme1-dec}

A correction round for this scheme consists of GKP qubit error correction (EC), followed by qubit stabilizer EC [see Fig. \ref{fig:decoding}(b)]. This procedure is the most widely used among the three we consider, and we refer the reader to Refs. \cite{terhal19,noh2020fault,larsen21fault,noh22surface,Hanggli_biased,color_gkp,xzzx_gkp} for more details.

\textbf{Layer 1: GKP EC.}
The first step is to measure the GKP stabilizers of each physical GKP qubit, i.e., $e^{i2\sqrt{\pi}\hat{q}_j}$ and $e^{-i2\sqrt{\pi}\hat{p}_j}$ for \(j\in\{1,2,\cdots,n\}\). Measuring these yields a \(2n\)-dimensional syndrome vector,
\begin{equation}\label{eq:syndromes-scheme-I}
    \vec{z}=(z_{1,q},\cdots,z_{n,q},z_{1,p},\cdots,z_{n,p})^T~,
\end{equation} 
consisting of phase-like GKP syndromes $z_{j,q}, z_{j,p}\in[-\sqrt{\pi}/2,\sqrt{\pi}/2)$. 

There are many possible noise vectors \(\vec\xi\) in Eq.~\eqref{eq:noise_vector} that yield a particular syndrome vector, and the next step is to deduce the most likely one (i.e., apply \textit{maximum-likelihood decoding}).
Each noise vector component \(\xi_j\) can be expressed as a sum of some integer multiple of \(\sqrt\pi\) and a remainder term,
\begin{equation}
    \xi_j = m\sqrt{\pi} + R_{\sqrt{\pi}}(\xi_j)~,
\end{equation}
where \(m\in\mathbb Z\), where we use the \textit{remainder function}
\begin{equation}\label{eq:remainder}
    R_s(x)=x-s\left[\frac{x}{s}\right]~,
\end{equation}
and where $[x]$ rounds \(x\) to the nearest integer. The remainder term is precisely what is extracted via syndrome measurements, \(z_j = R_{\sqrt{\pi}}(\xi_j)\), and the shortest deduced displacement vector is thus
\begin{equation}
  \vec{\xi}_* = \vec{z} = R_{\sqrt{\pi}}(\vec{\xi}).
\end{equation}

In order to correct, we apply a collective displacement by $-\vec{\xi}_*$, yielding the residual noise vector
\begin{eqs}\label{eq:gkp_correction}
    \vec{\xi}^\prime = \vec{\xi} - \vec{\xi}_* = \vec{\xi} - R_{\sqrt{\pi}}(\vec{\xi})
\end{eqs}
and completing the first layer of correction.

\textbf{Layer 2: qubit stabilizer EC.} 
The first layer has recovered the logical information back to the logical GKP subspace of each mode. 
The residual noise vector \(\vec{\xi}^\prime\) imposes a logical Pauli error on the inner stabilizer code.
To extract this error, we need to measure the GKP form of the stabilizers of the outer $\llbracket n,k,d\rrbracket $ code. 
These are constructed using tensor products of GKP displacements $X_{j}=e^{-i\sqrt{\pi}\hat p_j}$ and $Z_{j}=e^{i\sqrt{\pi}\hat q_j}$, which act as logical Pauli operators on the inner GKP code of mode \(j\). 
For example, for the $\llbracket 5,1,3\rrbracket $ stabilizer outer code, if we want to extract the syndrome corresponding to the check operator $IXZZX$, we need to measure $\exp[i\sqrt{\pi}(-\hat{p}_2+\hat{q}_3+\hat{q}_4-\hat{p}_5)]$. 
Since the noise vector consists of integer multiples of \(\sqrt{\pi}\), each measured syndrome value can only be \(\pm 1\).

After extracting the binary syndromes, one needs to determine the error based on the syndromes, and apply a GKP Pauli correction operation. These steps complete the second layer of correction.

After the above two-layer round of correction, one can apply the decoding circuit \(\enc^\dagger\) if one wants to obtain the logical information, or repeat the round if one wants to further preserve the information.

\subsection{Scheme II: GKP-analog decoding}\label{subsubsec:scheme2-dec}

A correction round for this scheme consists of nullifier-based mode-into-mode error correction (EC), followed by conventional GKP qubit EC [see Fig. \ref{fig:decoding}(c)].

\textbf{Layer 1: analog EC.}
The first step in this scheme is to measure the nullifiers \(\enc \hat{q}_j \enc^\dagger\) for \(j\in\{k+1,k+2,\cdots,n\}\) of the outer analog stabilizer code. This yields real-valued nullifier syndromes $z_j$ for \(j\in\{1,2,\cdots,n-k\}\), collectively denoted by the \((n-k)\)-dimensional syndrome vector $\vec{z}$. 

Nullifiers are related to the unencoded auxiliary mode position operators by the \((n-k)\)-by-\(2n\)-dimensional matrix \(A_1\) from Eq. (\ref{eq:A123}), and the syndrome vector is similarly related to the noise vector \(\vec{\xi}\) in Eq.~\eqref{eq:noise_vector} by the equation \(\vec{z}=A_1 \vec{\xi}\). 

Since \(A_1\) is rectangular, several different noise vectors can yield the same syndrome vector. The principle of maximum-likelihood decoding implies that we pick the shortest \(\xi\) that is compatible with the syndromes,
\begin{equation}\label{eq:opt1}
    \vec{\xi_*}=\argmin_{A_1 \vec{\xi}=\vec{z}} \|\vec{\xi}\|,
\end{equation}
where \(\|\vec{v}\|=\sqrt{\vec{v}\cdot\vec{v}}\) is the Hilbert-Schmidt norm of \(\vec{v}\). 

Finding the shortest compatible noise vector turns out to be a standard minimization problem
\citep{bertsimas1997introduction}, whose solution is given by
\begin{eqs}\label{eq:ml-error}
\vec{\xi_*}=A_1^{T}(A_1A_1^{T})^{-1}\vec{z}.
\end{eqs}
Above, \(A_1^{T}\) is the transpose of \(A_1\), and \(A_1 A_1^{T}\) is a \((n-k)\)-dimensional square matrix that is invertible since all nullifier measurements are linearly independent. Note that \(A_1^{T}(A_1A_1^{T})^{-1}\) is also called the right Moore-Penrose pseudoinverse of \(A_1\). 

The first layer of error correction is then performed by deducting the estimated noise vector $\vec{\xi_*}$ from the system. The updated quadrature noise vector $\vec{\xi}^{\prime}$ takes the form
\begin{equation}\label{eq:post-EC-error}
    \vec{\xi}^{\prime}=\vec{\xi}-\vec{\xi_*} = P_{A_1}^{\perp}\vec{\xi}~,
\end{equation}
where we use the formula for the projection onto the kernel of a matrix \(M\),
\begin{eqs}\label{eq:proj}
    P_{M}^{\perp}=I-M^T(M M^T)^{-1} M~,
\end{eqs} 
satisfying \(M P_M^\perp = 0\) and \(P_M^\perp M^T = 0\).
In other words, this layer of correction applies shifts to the nullifier quadratures such that the nullifier expectation values are reset to zero. 

The above layer of correction yields a shorter residual noise vector: \(\vec{\xi}^\prime\) is a shorter than \(\vec \xi\) since \(P_{A_1}^{\perp}\) is a projection. 
While this layer cannot decrease the variance of the logical-mode quadrature noise~\cite{terhal19,Noh_many_oscillator}, we are interested in encoding logical qubits in said modes and thus proceed to the second layer of correction.

\textbf{Layer 2: GKP EC.}
The above analog correction procedure has mapped the outer mode-into-mode encoding back into the ``logical'' \(k\)-mode space defined as the collective 0-eigenvalue subspace of all \(n-k\) nullifiers. The next layer consists of detecting and correcting logical errors of the \(k\) GKP qubits encoded in the logical mode space. Such errors are caused by the residual noise vector \(P_{A_1}^{\perp}\vec{\xi}\) in Eq.~\eqref{eq:post-EC-error}.

The \(2k\) check operators measured in this round are GKP stabilizers of the first \(k\) modes conjugated by the encoding unitary \(\enc\), i.e., $\enc e^{i2\sqrt{\pi}\hat q_j} \enc^\dagger$ and $\enc e^{-i2\sqrt{\pi}\hat p_j} \enc^\dagger$ for \(j\in\{1,2,\cdots,k\}\). 
Measuring these yields a \(2k\)-dimensional syndrome vector,
\begin{equation}\label{eq:syndromes}
    \vec{z}^\prime=(z_{1,q}^\prime,...,z_{k,q}^\prime,z_{1,p}^\prime,...,z_{k,p}^\prime)^T~,
\end{equation} 
consisting of GKP syndromes $z_{j,q}^\prime, z_{j,p}^\prime\in[-\sqrt{\pi}/2,\sqrt{\pi}/2)$. 

The syndrome vector can equivalently be represented as the remainder of the residual noise vector \(P_{A_1}^{\perp}\vec{\xi}\) in Eq.~\eqref{eq:post-EC-error} encoded into the GKP logical space via the submatrix \(A_2\) in Eq.~\eqref{eq:A123} of the symplectic matrix \(\aenc\) corresponding to \(\enc\),
\begin{equation}\label{eq:gkpz}
 \vec{z}^\prime=R_{\sqrt{\pi}}(A_2P_{A_1}^{\perp}\vec{\xi})~.
\end{equation}
The remainder function $R_{\sqrt{\pi}}$, applied to each entry of the vector in the argument, ensures that only the modular quadrature information is extracted from the processed noise vector.

Applying the maximum likelihood principle, we need to find the most probable error vector after the first error correction, $\vec{\xi}^\prime_{*}$, that is consistent with the syndromes, $A_2\vec{\xi}^\prime_{*}=\vec{z}^\prime$. However, since the entries of the residual noise vector Eq.~\eqref{eq:post-EC-error} are correlated, the most likely error vector cannot be calculated via minimizing the norm of $\vec{\xi}^\prime$. Rather, as $P_{A_1}^{\perp}$ is a deterministic linear matrix, the most probable $\vec{\xi}^\prime$ should come from the most probable $\vec{\xi}$. This yields the optimization problem\footnote{\label{ft:opt} There is a caveat in the minimization constraint in Eq. \eqref{eq:optscheme2layer2}. That is whether to view $\vec{z}$ as a constant vector set by the measurement outcome in the first layer EC, or substitute $\vec{z}=A_1\vec{\xi}$ into the constraint. In Eq. \eqref{eq:optscheme2layer2}, we choose the latter. If we adopt the former point of view, $\vec{\xi}^{\prime}$ now becomes $\vec{\xi}-A_1^T(A_1A_1^T)^{-1}\vec{z}$, and the optimization problem becomes
\[
    \vec{\xi}^{\prime}_*=\argmin_{A_2(\vec{\xi}-A_1^T(A_1A_1^T)^{-1}\vec{z})=\vec{z}^{\prime}} \|\vec{\xi}\|-A_1^T(A_1A_1^T)^{-1}\vec{z}.
\]
The result is $ \vec{\xi}_{*}^{\prime}=A_2^T(A_2A_2^T)^{-1}\vec{z}^{\prime}-P_{A_2}^{\perp}A_1^T(A_1A_1^T)^{-1}\vec{z}$. This $\vec{\xi}_{*}^{\prime}$ is different from that in Eq. \eqref{eq:scheme2postlayer2}. However, if we calculate the final errors on the logical mode, it becomes $ \vec{\xi}_{\text{final},l}=A_2(\vec{\xi}^\prime-\vec{\xi}_{*}^{\prime})=A_2(\vec{\xi}-A_1^T(A_1A_1^T)^{-1}\vec{z})-\vec{z}^{\prime}$. In the final step we substitute in $\vec{z}=A_1\vec{\xi}$, then $\vec{\xi}_{\text{final},l}=A_2 P_{A_1}^{\perp}\vec{\xi}-\vec{z}^{\prime}$, which is the same result as obtained in Eq. \eqref{scheme2finalerror}. This shows that using either point of view we get the same result for the final logical mode errors.}
\begin{equation}\label{eq:optscheme2layer2}
    \vec{\xi}^{\prime}_*=P_{A_1}^{\perp}\left(\argmin_{A_2P_{A_1}^{\perp} \vec{\xi}=\vec{z}^{\prime}} \|\vec{\xi}\|\right)
\end{equation}
for the GKP layer of correction.
The above optimization is solved in the same way as Eq. (\ref{eq:opt1}), yielding
\begin{equation}\label{eq:scheme2postlayer2}
    \vec{\xi}_{*}^{\prime}=P_{A_{1}}^{\perp}(A_{2}P_{A_{1}}^{\perp})^{T}\left((A_{2}P_{A_{1}}^{\perp})(A_{2}P_{A_{1}}^{\perp})^{T}\right)^{-1}\vec{z}^{\prime}\,.
\end{equation}

After implementing the above as the recovery displacement for this second layer, we apply \(\aenc\) (the decoding map in the Heisenberg picture).
The final residual noise vector is
\begin{eqs}\label{eq:finalerror}
    \vec{\xi}_{\text{final}} = \aenc(\vec{\xi}^\prime-\vec{\xi}_{*}^{\prime})~.
\end{eqs}
If we want to focus on the errors of the first \(k\) ``logical'' modes that house the GKP qubits in Fig. \ref{fig:encoding-overview}(c), we can instead decode using the submatrix \(A_2\) in Eq.~\eqref{eq:A123} and remainder function \textcolor{black}{$R$} in Eq.\eqref{eq:remainder} to obtain
\begin{align}\label{scheme2finalerror}
    \vec{\xi}_{\text{final},l}&=A_2(\vec{\xi}^\prime-\vec{\xi}_{*}^{\prime})\\&=A_2 P_{A_{1}}^{\perp}\vec{\xi}-R_{\sqrt{\pi}}(A_{2}P_{A_{1}}^{\perp}\vec{\xi})~.\nonumber
\end{align}

\textbf{Simplifying cases.}
The remainder operation in the noise vector expression in Eq.~\eqref{scheme2finalerror} can be removed in the low-noise case, since \(R_{\sqrt{\pi}}(\xi)=\xi\) for sufficiently small \(\xi\). 
In that case, \(\vec{\xi}-\vec{\xi}_{*}^{\prime} = P_{\tilde A}^\perp \vec{\xi}\), where \(P_{\tilde A}^\perp\) in Eq.~\eqref{eq:proj} is the projection onto the kernel of the block matrix \(\tilde{A} = \left(\begin{smallmatrix} A_1 \\ A_2 \end{smallmatrix}\right)\).
This means that, if we disregard the caveat that the logical GKP syndrome is measured modulo $\sqrt{\pi}$, we can combine the two layers of the error correction into one and obtain the same result. 

Another simplifying case is the condition $A_1 A_2^T=0$, which means that the measurements of $A_1$ and $A_2$ are in two orthogonal hyper-planes. In that case, the optimization problems for the two layers become independent, $R_{\sqrt{\pi}}(A_2P_{A_1}^{\perp}\vec{\xi})=R_{\sqrt{\pi}}(A_2\vec{\xi})$, and the order in which the corrections for the two layers are applied does not matter.

\subsection{Scheme III: GKP-stabilizer decoding}\label{subsubsec:scheme3-dec}

A correction round for this scheme consists of canonical-GKP stabilizer mode-into-mode error correction \cite{Noh_many_oscillator}, followed by conventional GKP qubit error correction [see Fig. \ref{fig:decoding}(d)]. We modify the canonical-GKP decoding procedure such that the unitary \(\enc^\dagger\) is applied last (whereas it originally preceded the canonical-GKP syndrome measurements in Ref.~\cite{Noh_many_oscillator}) in order to make this scheme consistent with schemes I and II, and in order to demonstrate how multiple rounds of error correction can be performed. 

Our modification also allows the measurement of other sets of GKP stabilizer generators. The generators are determined by the matrix \(A_3\) in Eq.~\eqref{eq:A123}, but there are many such possible matrices, all related to each other by a linear transformation of rows. This means we can optimize the generator matrix $A_3^\prime$ such that each row has the lowest possible norm, which, according to our observations, improves decoder performance.

\textbf{Layer 1: GKP-stabilizer EC.}
The first layer of this scheme is similar to the second layer of scheme II in that both measure GKP-type stabilizers. Here, one starts by measuring the \(2(n-k)\) canonical-GKP stabilizers associated with the \(n-k\) auxiliary modes, namely, the operators \(\enc e^{i\sqrt{2\pi}\hat q_j} \enc^\dagger\) and \(\enc e^{-i\sqrt{2\pi}\hat p_j}\enc^\dagger\) for \(j\in\{k+1,k+2,\cdots,n\}\). Measuring these yields a \(2(n-k)\)-dimensional vector \(\vec{z}\) of canonical-GKP syndromes $z_{j,q}, z_{j,p}\in[-\sqrt{\pi/2},\sqrt{\pi/2})$ [cf. Eq. (\ref{eq:syndromes})].

The syndrome vector can equivalently be represented as the remainder of the noise vector \(\vec{\xi}\) in Eq.~\eqref{eq:noise_vector} encoded into the GKP logical space via the submatrix \(A_3\) in Eq.~\eqref{eq:A123} of the symplectic matrix \(\aenc\) corresponding to \(\enc\),
\begin{equation}\label{eq:gkpz}
 \vec{z} = R_{\sqrt{2\pi}}(A_3\vec{\xi})~.
\end{equation}
The remainder function $R_{\sqrt{2\pi}}$, as defined in Eq.\eqref{eq:remainder}, now modulo \(\sqrt{2\pi}\) because canonical GKP states are used for auxiliary modes, once again extracts only the modular quadrature information.

We once again pick the shortest $\vec{\xi}$ that is compatible with the syndromes \(\vec{z}\), which yields an optimization identical to that from Eqs. (\ref{eq:ml-error}-\ref{eq:post-EC-error}),
\begin{equation}\label{eq:opt3}
    \vec{\xi}_*=\argmin_{A_3 \vec{\xi}=\vec{z}} \|\vec{\xi}\| = A_3^T(A_3A_3^T)^{-1}R_{\sqrt{2\pi}}(A_3\vec{\xi})~.
\end{equation}
Subtracting this most likely correction from the initial noise vector completes this layer and yields
\begin{equation}\label{eq:post-EC-error-III}
  \vec{\xi}^\prime=\vec{\xi}-\vec{\xi}_*=\vec{\xi}-A_3^T(A_3A_3^T)^{-1}R_{\sqrt{2\pi}}(A_{3}\vec{\xi})~.
\end{equation}
Note that the presence of the remainder function \(R\) obstructs us from the simplifications we were able to do for the first layer of scheme II [cf. Eq. (\ref{eq:post-EC-error})].

\textbf{Layer 2: GKP EC.}
We proceed to measure the \(2k\) GKP stabilizers of the first \(k\) modes conjugated by the encoding unitary \(\enc\), recovering the same \(2k\)-dimensional GKP syndrome vector in Eq.~\eqref{eq:syndromes} as that in layer 2 of scheme II.

Proceeding analogously to scheme II, we express the syndrome vector in terms of the layer-one residual error vector \(\vec{\xi}^\prime\) in Eq.~\eqref{eq:post-EC-error-III} encoded into the logical modes using the rectangular matrix \(A_2\) in Eq.~\eqref{eq:A123} and restricted to only its modular components via the GKP-qubit remainder operation,
\begin{equation}
    \vec{z}^{\prime}=R_{\sqrt{\pi}}A_{2}\left(\vec{\xi}-A_{3}^{T}(A_{3}A_{3}^{T})^{-1}R_{\sqrt{2\pi}}(A_{3}\vec{\xi})\right)~.
\end{equation}
Above, the shorthand notation \(R_{\sqrt{\pi}}A_{2}(\vec{v})=R_{\sqrt{\pi}}(A_{2}(\vec{v}))\).

The maximum likelihood problem for this layer is to find the most probable original error configuration \(\vec\xi\) that is compatible with \(\vec{z}^\prime\). This time, however, the optimization is not linear because the map from \(\vec\xi\) to \(\vec{\xi}^\prime\) in Eq.~\eqref{eq:post-EC-error-III} is not linear due to the remainder function \(R_{\sqrt{2\pi}}\).

We proceed with a related \textit{linear} optimization problem
\begin{equation}
    \vec{\xi}_{*}^{\prime}=\argmin_{\vec{z}^{\prime}=A_{2}(\vec{\xi}-A_3^T(A_3A_3^T)^{-1}\vec{z})}\|\vec{\xi}\|~,
\end{equation}
where \(\vec{z}\) is a constant fixed by the measurement outcomes.
The solution is
\begin{equation}
   \vec{\xi}_{*}^{\prime}=A_2^T(A_2A_2^T)^{-1}(\vec{z}^{\prime}+A_2A_3^T(A_3A_3^T)^{-1}\vec{z})~.
\end{equation}

The above optimization is different from the nonlinear case where \(\vec{z}\) is input as a function of \(\vec\xi\). However, the two types of optimizations were tried for Scheme II (see foonote \ref{ft:opt}) and, despite yielding different outcomes, still corresponded to the same residual noise vector \(\vec{\xi}_{\text{final},l}\) on the \(k\) logical modes. We thus have some evidence to believe that this optimization may not be too far off from the true nonlinear one.

Using the above result for correcting layer-two displacement yields the final residual noise vector \(\vec{\xi}_{\text{final}}\) in Eq.~\eqref{eq:finalerror} after the two layers of scheme III and the decoding operation \(\aenc\). 
The logical-mode residual noise subvector is
\begin{align}\label{eq:finalerror3}
&\vec{\xi}_{\text{final},l}=A_{2}(\vec{\xi}^\prime-\vec{\xi}_{*}^{\prime})\\=&A_{2}\vec{\xi}-A_{2}A_{3}^{T}(A_{3}A_{3}^{T})^{-1}\vec{z}-R_{\sqrt{\pi}}(A_{2}(\vec{\xi}-A_{3}^{T}(A_{3}A_{3}^{T})^{-1}\vec{z})).\nonumber
\end{align}

One can view both Scheme II and III as examples from a family of schemes whose initial auxiliary modes are in GKP states with stabilizers $e^{i\sqrt{2\pi/\alpha}\hat{q}_i}$ and $e^{i\sqrt{2\pi \alpha}\hat{p}_i}$. The period in $q_i$ is $\sqrt{2\pi \alpha}$ and that in $p_i$ is $\sqrt{2\pi/\alpha}$. Scheme III corresponds to the case $\alpha=1$, while Scheme II can be viewed as the limiting case $\alpha\to\infty$. (Scheme I is $\alpha=2$.) We apply $R_{\sqrt{2\pi \alpha}}$ to the $q_i$ syndrome measurement and $R_{\sqrt{2\pi/\alpha}}$ to that of $p_i$. When $\alpha\to\infty$, the period in $q_i$ goes to infinity so no rounding is needed in the first layer EC of Scheme II and $\vec{z}=A_1\vec{\xi}$ . But the period in $p_i$ becomes infinitesimal. Applying $R_{\sqrt{2\pi/\alpha}}$ will round any measurement result to $0$. So no information can be subtracted from $p_i$ measurement and they are omitted in Scheme II. Correspondingly, the $A_3$ matrix in Scheme III reduces to $A_1$ in Scheme II.

\subsection{Calculating logical error rates}\label{sec:calc_logical_rate}

In order to compare the above schemes, one can calculate the logical error probabilities induced by the residual noise vector \(\vec{\xi}_{\text{final}}\) on the logical GKP qubits housed in the \(k\) logical modes [see Fig. \ref{fig:encoding-overview}(c)]. Such a calculation is done somewhat differently in scheme I than in schemes II and III, due to the latter two having continuous-variable outer codes.

Let $\mathfrak{p}(\vec{\xi})$ denote an arbitrary distribution for the initial noise vector $\vec{\xi}$ in Eq.~\eqref{eq:noise_vector}. Recall that the noise vector tracks a particular instance of random quadrature displacements in Eq.~\eqref{eq:displacement_noise}, which are usually independently distributed according to a Gaussian distribution with mean zero and fixed standard deviation (in which case $\mathfrak{p}(\vec{\xi})\equiv\prod_{j=1}^{2n} \mathfrak{p}(\xi_{j})$ with \(\mathfrak{p}(\xi_{j})\) a Gaussian distribution). We emphasize that the analysis of this subsection is independent of the choice of distribution.

\textbf{Schemes II \& III.} After two layers of EC and application
of $\enc^{\dagger}$, the initial noise vector $\vec{\xi}$ transforms
as
\begin{equation}
\vec{\xi}\to\vec{\xi}_{\text{final}}\equiv\vec{f}(\vec{\xi})\,,
\end{equation}
where \(\vec{\xi}_{\text{final}}\) in Eq.~\eqref{eq:finalerror} is the residual \(2n\)-dimensional logical-mode noise vector after two layers of correction and a decoding map for either schemes II or III, and the map from \(\vec\xi\) to this vector is represented by the vector-valued function $\vec{f}$. The first \(2k\) components of this vector are in Eq. \eqref{scheme2finalerror} and \eqref{eq:finalerror3} for schemes II and III, respectively.

On the logical-mode level, the probability of a displacement by $\eta$
of a logical-mode quadrature $j\in\{1,2,\cdots,2k\}$ is an integral over
contributions from all shifts $\vec{\xi}$ that are compatible with
the final outcome being the $j$th component of $\vec{f}(\vec{\xi}\,)$,
\begin{equation}
\text{Pr}_{j}(\eta\,|\,\vec{f}\,)=\int_{-\infty}^{\infty}d^{2n}\xi\,\mathfrak{p}(\vec{\xi}\,)\,\delta\left(\eta-f_{j}(\vec\xi)\right),\label{eq:disp}
\end{equation}
where $d^{2n}\xi\equiv d\xi_{q_1}...d\xi_{q_n}d\xi_{p_1}...d\xi_{p_n}$ is the integration measure on all quadratures.
This is the logical-mode displacement distribution associated
with $\vec f$.

On the logical-qubit level, GKP error correction succeeds whenever
$\eta-R_{\sqrt{\pi}}(\eta)$ is an even integer. In other words, the
probability $p_{j}^{\text{no error}}$ of successful correction of
the $j$th quadrature is the integral of the displacement distribution
in Eq.~\eqref{eq:disp} over a set of intervals centered at even multiples
of $\sqrt{\pi}$ (see, e.g., \cite[Eq. (10)]{noh22surface}),
\begin{equation}\label{eq:integrals}
p_{j}^{\text{no error}}=\sum_{m\in\mathbb{Z}}\int_{\left(2m-\frac{1}{2}\right)\sqrt{\pi}}^{\left(2m+\frac{1}{2}\right)\sqrt{\pi}}d\eta\,\text{Pr}_{j}(\eta\,|\,\vec{f}\,)\,.
\end{equation}
For few-mode codes, these integrals can often be done analytically.

Logical errors result when at least one $p_{j}^{\text{no error}}$ is
nonzero.
For a code with \(k\) logical qubits and uncorrelated displacement errors,
the logical error probability is the complement of the product of no-error probabilities of all of the quadratures,
\begin{equation}
p^{\text{logical error}}=1-\prod_{j=1}^{2n}p_{j}^{\text{no error}}\,.
\end{equation}
For correlated noise, this should become a lower bound.

\textbf{Scheme I.}
After the first layer of EC for this scheme, each mode can be readily treated as a GKP qubit, encoded in a mode that in turn is made up of a position and a momentum quadrature. As this is the first layer, there is no additional processing of the noise vector, meaning that error probabilities can be calculated as a special instance of those of schemes II and III, but with the processing function \(\vec{f}\) being identity.
The respective probabilities of no bit- or no phase-flips for a GKP qubit encoded in mode \(j\) are otherwise analogous to Eq. (\ref{eq:integrals}),
\begin{eqs}
p_j^{\text{no \(Z\) error}}&=\sum_{m\in\mathbb{Z}}\int_{\left(2m-\frac{1}{2}\right)\sqrt{\pi}}^{\left(2m+\frac{1}{2}\right)\sqrt{\pi}}d\eta\,\text{Pr}_{j}(\eta\,|\,1),\\
p_j^{\text{no \(X\) error}}&=\sum_{m\in\mathbb{Z}}\int_{\left(2m-\frac{1}{2}\right)\sqrt{\pi}}^{\left(2m+\frac{1}{2}\right)\sqrt{\pi}}d\eta\,\text{Pr}_{j+n}(\eta\,|\,1)~.
\end{eqs}

With these intrinsic $X$ and $Z$ error probabilities of physical GKP qubits, one can calculate the logical error probability just as in the usual qubit stabilizer codes. The final expression depends on the stabilizer code we choose.

\textbf{Simplified error-rate calculations.}
We have discovered a simplification in calculating the logical error rate of a two-layer round of error correction for schemes II and III. Namely, calculating the rates \(p_{j}^{\text{no error}}\) does not require the layer-two GKP recovery operation to be present in the function \(\vec{f}\).

In Appendix \ref{appendix:logicalGKPEC}, we show that, if \(\vec{f}(\vec{\xi})\) defines the resulting error vector after one layer of correction, the second layer modifies as \(\vec{f}(\vec{\xi})\to\vec{f}(\vec{\xi})-R_{\sqrt{\pi}}(\vec{f}(\vec{\xi}))\), identical to the bare GKP correction scheme in Eq. (\ref{eq:gkp_correction}).
We combine this with the fact that each of the \(f\)-dependent integrals in Eq. (\ref{eq:integrals}) is invariant under \(f\to f-R_{\sqrt{\pi}}f\),
\begin{equation}
    \int_{\left(2m-\frac{1}{2}\right)\sqrt{\pi}}^{\left(2m+\frac{1}{2}\right)\sqrt{\pi}}\!\!\!\!d\eta\text{Pr}_{j}(\eta\,|\,\vec{f}\,)\,=\int_{\left(2m-\frac{1}{2}\right)\sqrt{\pi}}^{\left(2m+\frac{1}{2}\right)\sqrt{\pi}}\!\!\!\!d\eta\text{Pr}_{j}(\eta\,|\,\vec{f}-R_{\sqrt{\pi}}\vec{f}\,)\,,
\end{equation}
to show that explicit GKP recovery is not necessary to calculate the logical-mode, and therefore logical-qubit, error rates. 

GKP recovery, of course, still has to be performed to yield a logical qubit encoding governed by the aforementioned error rates.
The benefit of GKP correction still exists in the error-rate calculation because the displacement distribution is integrated over a union of segments comprising half of the real line, which is the correctable region of a GKP encoding.

\begin{figure}[h]
    \centering
    \includegraphics[width=0.45\textwidth]{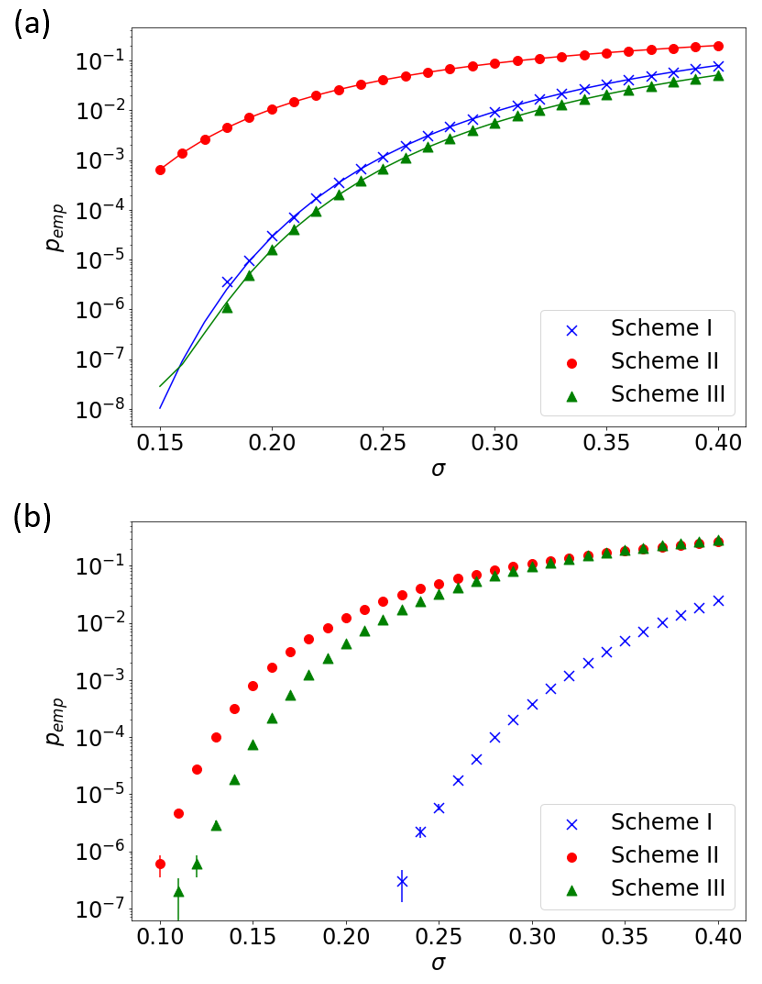}
    \caption{\textsc{Numerical simulations:} (a) is the numerical simulation of 3-qubit repetition codes in three different schemes; (b) is the numerical simulation of $\llbracket 5,1,3\rrbracket $ in three different schemes, according to the prescription in Appendix.~\ref{appendix:numerics}. The horizontal axis $\sigma$ represents the variance of the Gaussian displacement channel as discussed in Sec.~\ref{sec: displacement_error}. The vertical axis $p_{\text{emp}}$ represents the empirical logical error rates that are numerically calculated by the method in Appendix~\ref{appendix:numerics}.}  Each data point is obtained by averaging $10^7$ samples.
    \label{fig:repetition_513_numerics}
\end{figure}

\section{Examples}\label{sec:example}
We numerically benchmark two sets of examples of the three schemes from Fig. \ref{fig:encoding-overview}, one using the repetition code and its bosonic variants, and another using the five-qubit code and its variants. Details of our Monte Carlo sampling of quadrature noise and its decoding are given in Appendix.~\ref{appendix:numerics}.

\subsection{Repetition-code based comparison}\label{sec:repetition}

The repetition code \cite{peres85reversible} is an important example in both classical and quantum error-correcting codes. There are previous studies on the scheme I concatenated repetition-GKP codes \cite{fukui17} and $\llbracket 4,2,2\rrbracket$ -GKP codes \cite{fukui18}. 

The three encoding schemes for the repetition-code example are described by the following map [cf. Fig. \ref{fig:encoding-overview}(c)],
\begin{eqs}
\text{I. repetition-GKP:}\quad &\mathbb{Z}_2 \xrightarrow[]{\text{repetition}} \mathbb{Z}_2^{\otimes 3} \xrightarrow[]{\text{GKP}} \mathbb{R}^{\otimes 3}\\
\text{II. analog-repetition:}\quad &\mathbb{Z}_2 \xrightarrow[]{\text{GKP}} \mathbb{R} \xrightarrow[]{\text{analog repetition}} \mathbb{R}^{\otimes 3}\\
\text{III. GKP-repetition:}\quad &\mathbb{Z}_2 \xrightarrow[]{\text{GKP}} \mathbb{R} \xrightarrow[]{\text{GKP-repetition}} \mathbb{R}^{\otimes 3}~.
\end{eqs}
Scheme I is a standard concatenated repetition-GKP code constructed by encoding each physical qubit into a GKP code. 
Scheme II  is the analog repetition code (see Appendix A of Ref. \cite{Noh_many_oscillator}) concatenated with a bona-fide GKP code.
Scheme III replaces the analog code with the GKP-repetition code \cite{Noh_many_oscillator}, which suppresses the variance of position error acting on logical mode to $\sigma^2/3$ if the error rate is sufficiently low, concatenated with a GKP qubit code. 

Encoding circuits for all three schemes are of similar form (see Fig. \ref{fig:encoding-overview}). The relevant Gaussian unitary \(\enc\) is
\begin{eqs}
\enc=\text{CNOT}_{1\rightarrow 3}\text{CNOT}_{1\rightarrow 2}~,
\end{eqs}
where our CNOT \footnote{The CV CNOT gate is $\text{CNOT}_{j\rightarrow k}=\exp[-i \hat{q}_j \hat{p}_k]$.} two-mode gate acts on quadratures as
\begin{eqs}\label{eq:cv_heisenberg}
\text{CNOT}_{j\rightarrow k} \hat{q}_j \text{CNOT}_{j \rightarrow k}^\dagger &=\hat{q}_j\\
\text{CNOT}_{j\rightarrow k} \hat{q}_k \text{CNOT}_{j \rightarrow k}^\dagger &=\hat{q}_k-\hat{q}_j\\
\text{CNOT}_{j\rightarrow k} \hat{p}_j \text{CNOT}_{j \rightarrow k}^\dagger &=\hat{p}_j+\hat{p}_k\\
\text{CNOT}_{j\rightarrow k} \hat{p}_k \text{CNOT}_{j \rightarrow k}^\dagger &=\hat{p}_k~.
\end{eqs}
The transformation of Hadamard and CZ gates\footnote{ CV Hadamard gate is $\exp[i\frac{\pi}{2} a^\dagger a]$ and CZ is chosen to be $H_k \text{CNOT}_{j\rightarrow k} H_k^\dagger$.} are

\begin{eqs}
H \hat{q} H^\dagger &=\hat{p}\\
H \hat{p} H^\dagger &=-\hat{q}\\
\text{CZ}_{j \rightarrow k} \hat{q}_j \text{CZ}_{j \rightarrow k}^\dagger &= \hat{q}_j,\\
\text{CZ}_{j \rightarrow k} \hat{q}_k \text{CZ}_{j \rightarrow k}^\dagger &= \hat{q}_k,\\
\text{CZ}_{j \rightarrow k} \hat{p}_j \text{CZ}_{j \rightarrow k}^\dagger &= \hat{p}_j-\hat{q}_k,\\
\text{CZ}_{j \rightarrow k} \hat{p}_k \text{CZ}_{j \rightarrow k}^\dagger &= \hat{p}_k-\hat{q}_j,
\end{eqs}
which will be used later.

The error-correction for scheme I (see Sec. \ref{subsubsec:scheme1-dec}) specializes to the following. First, we perform GKP error correction on individual GKP modes to eliminate small displacement errors acting on each mode. Then we measure repetition-code Pauli stabilizers $Z_1 Z_2$ and $Z_1 Z_3$ and correct GKP logical errors acting on individual modes. The total number of syndromes is 8.\footnote{We apply the minimum weight decoder in Scheme I as in usual qubit codes. For 3-repetition code, this is related to the maximal likelihood decoder in Schemes II and III as follows. Measuring the stabilizer $Z_1Z_2$ and $Z_1Z_3$ gives an outcome vector $\vec{z}$. Using maximal likelihood to deduce the most likely error is the same as in Eq.\eqref{eq:ml-error}. However, knowing that the errors after the first round should be an integer multiple of $\sqrt{\pi}$,  the most likely error should be $R_{\sqrt{\pi}}(A_1^{T}(A_1A_1^{T})^{-1}\vec{z})$. One can check that this gives the same prediction as the minimal weight decoder.}

For scheme II, before encoding, the first mode is a GKP logical state, and the rest of the modes are initialized in $\ket{\hat{q}=0}$. After the encoding, the code state is stabilized by GKP stabilizers and nullified by nullifiers simultaneously. The GKP stabilizers are $\exp[-i 2\sqrt{\pi}(\hat{p}_1+\hat{p}_2+\hat{p}_3)]$ and $\exp[i 2\sqrt{\pi}\hat{q}_1]$, and these square-roots yield logical Pauli-\(X,Z\) gates for the encoded qubit.
The nullifiers of the outer code are \(\hat{q}_2-\hat{q}_1\) and \(\hat{q}_3- \hat{q}_1\). The total number of syndromes is 4.

For scheme II error-correction procedure  (see Sec. \ref{subsubsec:scheme2-dec}), we first measure the nullifiers and obtain syndromes $\xi_{2,q}-\xi_{1,q}$ and $\xi_{3,q}-\xi_{1,q}$. We use a maximum-likelihood decoder to perform error correction, obtaining an error-corrected quadrature $\vec{\xi}'$ in Eq.~\eqref{eq:post-EC-error}.
We then measure the GKP stabilizers, which returns $R_{\sqrt{\pi}}(\xi_{1,p}'+\xi_{2,p}'+\xi_{3,p}')$ and $R_{\sqrt{\pi}}(\xi_{1,q}')$ [where \(R\) is the remainder function in Eq.~\eqref{eq:remainder}]. We proceed to do maximum-likelihood error correction based on syndromes diagnosing the residual noise vector $\vec{\xi}'$ after the first layer of correction.

In scheme III, we first measure the auxiliary canonical-GKP stabilizers, obtaining the four syndromes
\begin{eqs}
z_{2,q}&=R_{\sqrt{2\pi}} (\xi_{2,q}-\xi_{1,q}),\\ z_{3,q}&=R_{\sqrt{2\pi}} (\xi_{3,q}-\xi_{1,q}),\\ z_{2,p}&=R_{\sqrt{2\pi}} (\xi_{2,p}),\\ z_{3,p}&=R_{\sqrt{2\pi}} (\xi_{3,p}).
\end{eqs}
After maximum-likelihood error correction (see Sec.~\ref{subsubsec:scheme3-dec}), the residual noise vector in Eq.~\eqref{eq:finalerror3} acting on the logical mode is 
\begin{eqs}
\xi^\prime_{1,q}&=z_{1,q}+\frac{1}{3}(z_{2,q}+z_{3,q})= \frac{\xi_{1,q}+\xi_{2,q}+\xi_{3,q}}{3}\\
\xi^\prime_{1,p}&=z_{1,p}-z_{2,p}-z_{3,p}=\xi_{1,p}~,
\end{eqs}
where we have assumed that all components are less than \(\sqrt{2\pi}\) in order to remove the remainder function \(R_{\sqrt{2\pi}}\).
If we further assume that each initial fluctuation $\xi$ is an identical and independent random variable with zero mean and variance $\sigma^2$, the variances of the above errors are
\begin{eqs}
\xi^\prime_{1,q}&\sim \mathcal{N}(0,\frac{\sigma^2}{3})\\
\xi^\prime_{1,p}&\sim \mathcal{N}(0,\sigma^2).
\end{eqs}
Scheme III proceeds to do logical GKP syndrome measurement and correction, which brings the total number of syndromes up to 6.

The numerical comparison for the three schemes is shown in Fig.~\ref{fig:repetition_513_numerics} \textcolor{black}{(a)}. Scheme II performs the worst likely because of a no-go theorem for the first layer \cite{terhal19} (see also\cite{eisert02,niset09}), which states that a linear mode-into-mode code defined by a set of nullifiers can only squeeze the quadrature error but can never reduce noise on both quadratures. 

For the noise standard deviation $\sigma \geq 0.18$, scheme III outperforms scheme I, demonstrating the advantage of the canonical-GKP stabilizer formalism for small-scale qubit codes. Moreover, scheme III requires less resources, using 6 syndromes in contrast to the 8 syndromes of scheme I. This crossover behavior around $\sigma=0.18$ is also observed in 5-qubit and 7-qubit repetition codes (not shown in the figure).

To complement our numerical simulations, we analytically calculate distributions of position and momentum displacement errors acting on the logical mode after error correction, following Ref. \cite{Noh_many_oscillator}. Based on logical error distributions, we analytically obtain the logical error rates of different schemes, as shown in the part (a) of Fig.~\ref{fig:repetition_513_numerics}. They are in good agreement with Monte Carlo results, and reveal an eventual crossover between performance of schemes I and III. The calculation details are collected in Appendix \ref{appendix:analytic-3-rep}.

\begin{figure}[h]
    \centering
    \includegraphics[width=0.45\textwidth]{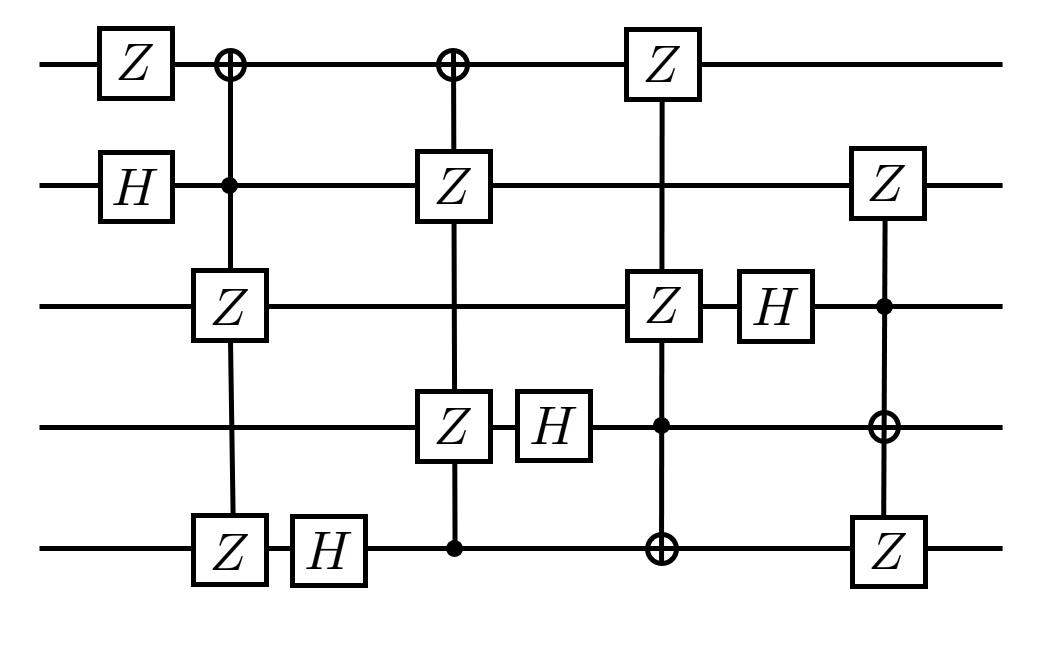}
    \caption{\textsc{Five-mode encoder \(\enc\) \cite{nakahara2008quantum}:} Here, $\text{CZ}_{i\rightarrow j}$ is chosen to be $H_j \text{CNOT}_{i \rightarrow j} H_j^\dagger$, and the top mode initially contains the logical information.}
    \label{fig:513_encoderv}
\end{figure}

\subsection{Five-qubit code based comparison}\label{sec:513}

The \(\llbracket 5,1,3\rrbracket \) qubit code \cite{laflamme_513} is the smallest qubit code to correct an arbitrary single-qubit Pauli error. Its continuous-variable version is studied in Ref.~\cite{braunstein98}.
Plugging in this code, the three schemes from Fig. \ref{fig:encoding-overview}(d) specialize to:
\begin{eqs}
\text{I. \(\llbracket 5,1,3\rrbracket \)-GKP:}\quad& \mathbb{Z}_2 \xrightarrow[]{\llbracket 5,1,3\rrbracket } \mathbb{Z}_2^{\otimes 5} \xrightarrow[]{\text{GKP}} \mathbb{R}^{\otimes 5}\\
\text{II. \(\llbracket 5,1,3\rrbracket \)-analog:}\quad& \mathbb{Z}_2 \xrightarrow[]{\text{GKP}} \mathbb{R} \xrightarrow[]{\llbracket 5,1,3\rrbracket _{\mathbb R}} \mathbb{R}^{\otimes 5}\\
\text{III. GKP-\(\llbracket 5,1,3\rrbracket \):}\quad& \mathbb{Z}_2 \xrightarrow[]{\text{GKP}} \mathbb{R} \xrightarrow[]{\text{GKP}-\llbracket 5,1,3\rrbracket } \mathbb{R}^{\otimes 5}.
\end{eqs}

For the error-correction part of scheme I, we first do error correction on each GKP qubit and then perform Pauli stabilizer measurements and error correction afterward.

For the error-correction part of scheme II, we first measure nullifiers, which are defined by the nullifier matrix $A_1$ [see Eq. (\ref{eq:A123})],
\begin{eqs}
A_1= \begin{bmatrix} -1&0 & -1 & 0 &0  &0  &0 &0 & 1&1\\
0& -1& 0& -1& 0& 1& 0& 0& 0& 1\\
0& -1& 0& 0& -1& 0& 0& 1& 1& 0\\
0& 0& -1& 0& -1& 1& 1& 0& 0& 0
\end{bmatrix}.\label{eq:513_nullifier}
\end{eqs}
The above matrix is equivalent to nullifiers given by qudit Pauli stabilizers \cite{chau97five}  $$\{Z^{-1}IZ^{-1}XX, XZ^{-1}IZ^{-1}X, IZ^{-1}XXZ^{-1}, XXZ^{-1}IZ^{-1}\}$$ 
(for $\mathbb{Z}_2$ qubits, $Z^{-1}$=$Z$).

In the second layer error correction of scheme II, we measure logical GKP stabilizers 
\begin{eqs}
\enc \hat{S}_q \enc^\dagger&= e^{i 2 \sqrt{\pi } (\hat{q}_1-\hat{q}_2-\hat{q}_3+\hat{q}_4-\hat{q}_5)}\\
\enc \hat{S}_p \enc^\dagger&= e^{-i 2 \sqrt{\pi } (\hat{p}_1+\hat{q}_3-\hat{q}_4)}~,
\end{eqs}
obtaining syndrome outcomes $R_{\sqrt{\pi}}(\xi_{1,q}'-\xi_{2,q}'-\xi_{3,q}'+\xi_{4,q}'-\xi_{5,q}')$ and $R_{\sqrt{\pi}} (\xi_{1,p}'+\xi_{3,q}'-\xi_{4,q}')$. The error corrected quadrature after the first layer is $\vec{\xi}'$ in Eq.~\eqref{eq:post-EC-error}.

For scheme III, if $\forall i\in \{1,2,...,n\},~z_{i,q}, z_{i,p} \ll \sqrt{2\pi}$, the error-corrected logical error quadratures after the first layer are  
\begin{eqs}
\xi^{\prime}_{1,q} &=\frac{1}{11}(5 \xi_{1,q}-2\xi_{1,p}-2\xi_{2,p}+3\xi_{3,p}+3\xi_{4,p}-2\xi_{5,p})\\
\xi^{\prime}_{1,p}&=\frac{1}{11}(-2 \xi_{1,q}+3\xi_{1,p}+3\xi_{2,p}+\xi_{3,p}+\xi_{4,p}+3\xi_{5,p}).
\end{eqs}
Ideally, $\xi^{\prime}_{1,q}\sim \mathcal{N}(0,\frac{5\sigma^2}{11})$, $\xi^{\prime}_{1,p}\sim \mathcal{N}(0,\frac{3\sigma^2}{11})$ at unambiguously distinguishable regime.

In numerical simulation (see Fig.~\ref{fig:repetition_513_numerics}(b)), we perform a Monte Carlo simulation (see Appendix.~\ref{appendix:numerics}) to study the logical error rate of three schemes. The result shows that scheme I performs the best while scheme II has the highest logical error rate in this regime.

We believe the difference in performance between schemes I and III is due to the following. Scheme I performs GKP-type modular quadrature correction on each mode first, which reduces individual quadrature noise before it is passed onto the next layer of correction. Scheme III, on the other hand, mixes said quadrature noise via \(\enc\) and implements modular correction on the \(n-k\) auxiliary modes only \textit{after} mixing. The mixing redistributes the initially uniform noise in an asymmetric way amongst the modes, opening up the possibility for noise from several quadratures to concentrate on one output mode, thereby increasing the resulting noise variance. This is not an issue if the variance of the resulting additive noise is much smaller than $\sqrt{2\pi}$, in which case the syndrome $z \mod \sqrt{2\pi}$ can be approximated by $z$. However, once the variance of $z$ is comparable to $\sqrt{2\pi}$, then $z \mod \sqrt{2\pi}$ can be quite different, leading to uncorrectable errors.

\subsection{Numerical simulation of Shor and Steane codes}\label{sec:shor_steane}

In this section, we use the same numerical method to study the error rates of the$\llbracket 9,1,3\rrbracket$ (Shor) code \cite{shor_code} and the $\llbracket 7,1,3 \rrbracket$ (Steane) code \cite{steane_code}. Both Shor and Steane codes can correct arbitrary single-qubit Pauli errors.

The numerical result is shown in Fig.~\ref{fig:shor_steane}. Scheme III performs the worst in both cases, compared to schemes I and II. This numerical result \textcolor{black}{is expected, according to the} conjecture on the ``error concentration" issue of GKP-stabilizer codes.

We first compare the performance of scheme III \textcolor{black}{with} the $\llbracket 7,1,3\rrbracket$ and $\l\llbracket 5,1,3 \rrbracket$ codes, because both of them only involve weight-4 stabilizers and can correct arbitrary single-qubit Pauli errors. We find that the logical error rate of scheme III are similar in $\llbracket 7,1,3 \rrbracket$ and $\llbracket 5,1,3 \rrbracket$ codes. 
\textcolor{black}{This} suggests that versions of scheme III with DV codes having similar stabilizer weight and distance will have similar logical error rate. 


Furthermore, we compare the logical error rate of scheme III \textcolor{black}{with the} $\llbracket 7,1,3\rrbracket$ and $\llbracket 9,1,3 \rrbracket$ codes. We find that scheme III has a lower logical error rate in $\llbracket 7,1,3\rrbracket$ codes than in the $\llbracket 9,1,3 \rrbracket$ code. Both $\llbracket 7,1,3\rrbracket$ and $\llbracket 9,1,3 \rrbracket$ can correct arbitrary single-qubit Pauli error. However, the maximum stabilizer weight of $\llbracket 7,1,3\rrbracket$ code is 4 while the maximum stabilizer weight of $\llbracket 9,1,3\rrbracket$ is 6. The correlation between stabilizer weights and logical error rates also consolidates our hypothesis of ``error concentration" that the performance of scheme III will be affected by high-weight stabilizers.

The logical error rates \textcolor{black}{for the} $\llbracket 5,1,3 \rrbracket$, $\llbracket 7,1,3 \rrbracket$, and $\llbracket 9,1,3 \rrbracket$ codes \textcolor{black}{using scheme III} are shown in Table.~\ref{table:distance-3_codes}. We found the logical error rates of scheme III $\llbracket 5,1,3 \rrbracket$ and $\llbracket 7,1,3 \rrbracket$ are in the same scale while the logical error rate of $\llbracket 9,1,3 \rrbracket$ is much greater than the other two codes. This fact also corroborates our previous error-concentration hypothesis.

\begin{table}[h]
\label{table:distance-3_codes}
\begin{tabular}{ |p{1cm}||p{2cm}|p{2cm}|p{2cm}|  }
 \hline
 $\sigma$ & $\llbracket 5,1,3 \rrbracket$ code & $\llbracket 7,1,3 \rrbracket$ code & $\llbracket 9,1,3 \rrbracket$ code\\
 \hline
 
0.15   & $7.410\times 10^{-5}$    & $1.810\times 10^{-5}$&   $3.765\times 10^{-4}  $\\
0.16   & $2.207\times 10^{-4}$    & $5.820\times 10^{-5}$&   $8.211 \times 10^{-4}  $\\
0.17   & $5.507\times 10^{-4}$    & $1.544\times 10^{-4}$&   $1.659 \times 10^{-3}  $\\
0.18   & $1.233\times 10^{-3}$    & $3.525\times 10^{-4}$&   $2.954 \times 10^{-3}  $\\
0.19   & $2.431\times 10^{-3}$    & $7.392\times 10^{-4}$&   $4.770\times 10^{-3}  $\\
0.20   & $4.393\times 10^{-3}$    & $1.374\times 10^{-3}$&   $7.319\times 10^{-3}$\\

 \hline
\end{tabular}
\caption{Scheme III yields comparable logical error rates in three different instances where a distance-three code is used.  Here we choose $0.15 \leq \sigma \leq 0.20$ where the logical error rates are not fully saturated in all three codes. }
\end{table}

\begin{figure}
    \centering
    \includegraphics[width=0.45\textwidth]{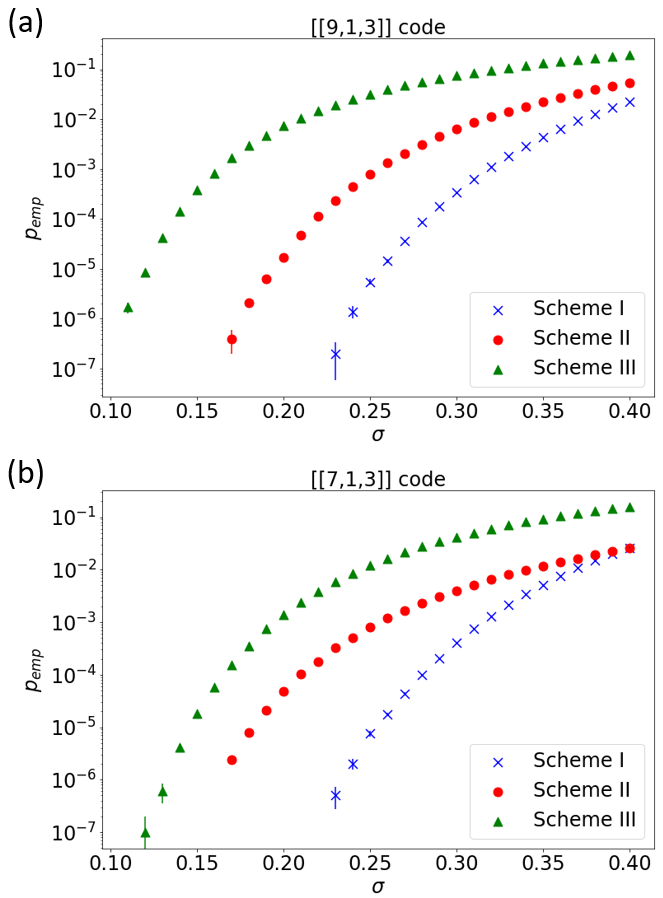}
    \caption{Numerical simulations of Shor and Steane code. The horizontal axis $\sigma$ represents the variance of the Gaussian displacement channel as discussed in Sec.~\ref{sec: displacement_error}. The vertical axis $p_{\text{emp}}$ represents the empirical logical error rates that are numerically calculated by the method in Appendix~\ref{appendix:numerics}. Each data point is obtained by averaging $10^7$ samples. }
    \label{fig:shor_steane}
\end{figure}

\section{UNBIASED GKP-REPETITION CODES}\label{sec:unbiased_code}

\begin{figure}[t]
    \raggedright
    \includegraphics[width=0.45\textwidth]{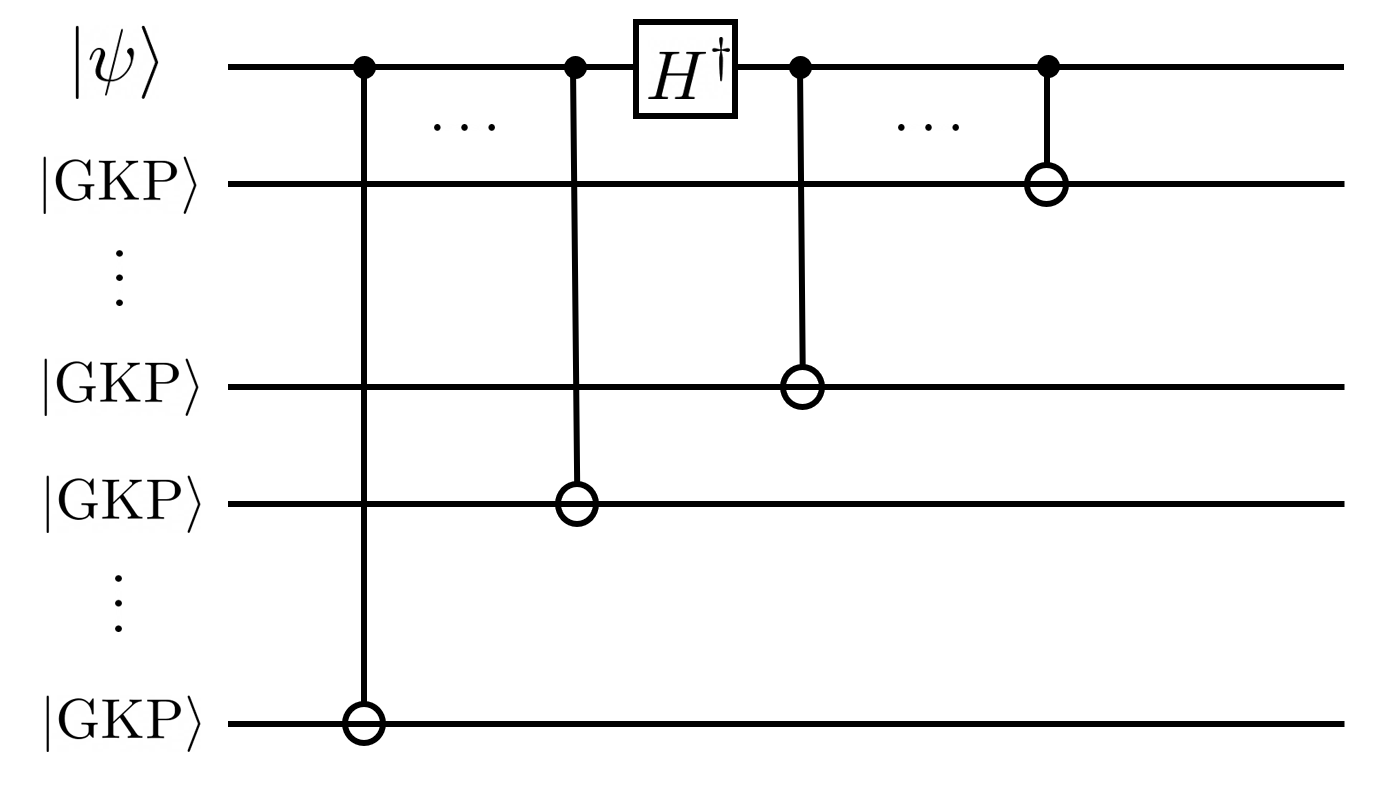}
    \caption{\textsc{Unbiased GKP-repetition encoder \(\enc\):} The controlled-$\bigominus$ denotes inverse-SUM gate.}
    \label{fig:unbiased_repetition_ecoder}
\end{figure}

We propose an unbiased GKP-repetition code that uses $2n$ auxiliary modes to simultaneously suppress the variance of both position and momentum errors of a single logical mode without extra quadrature squeezing. The code requires $2n+1$ modes
and allows the initial noise variance $\sigma^2$ to be suppressed to $\frac{\sigma^2}{n+1}$. Previous work [8] either requires squeezing to suppress
both quadratures, or achieves similar suppression in only one quadrature via the GKP-repetition code.

The encoding circuit \(\enc\) for this code is shown in Fig.~\ref{fig:unbiased_repetition_ecoder}, and we proceed to prove our stated claims using the formalism of Ref. \cite{Noh_many_oscillator}. 
We do not concatenate this code with a qubit code as was done in the formalism of Sec. \ref{sec:EC_dec}.

The main idea of canonical GKP-repetition codes is to propagate the position error $\xi_{0,q}$ to the position syndrome of auxiliary modes and then perform the maximum-likelihood estimation. Applying it to our case, we label the logical mode by 0; the rest of the \(2n\) modes are auxiliary modes.

Error syndromes $\vec{z}$ after the decoding circuit from Fig.~\ref{fig:unbiased_repetition_ecoder} can be written as
\begin{eqs}
\begin{cases}
&z_{0,q}=R_{\sqrt{\pi}}(\sum_{k=1}^{n}\xi_{k,p}-\xi_{0,p})\\
&z_{0,p}=R_{\sqrt{\pi}}(\xi_{0,q}-\sum_{k=n+1}^{2n}\xi_{k,p})\\
&z_{j,q}=R_{\sqrt{2\pi}}(\xi_{j,q}+\xi_{0,q}), \quad\quad\quad\quad\quad\quad 1\leq j \leq n\\
&z_{l,q}=R_{\sqrt{2\pi}}(\xi_{l,q}+\sum_{k=1}^n\xi_{k,p}-\xi_{0,p}),\quad n+1\leq l \leq 2n\\
&z_{s,p}=R_{\sqrt{2\pi}}(\xi_{s,p}), 1\leq s \leq 2n~.
\end{cases}
\end{eqs}
The logical Pauli operators for this code \textcolor{black}{are} 
\begin{eqs}\label{eq: stabilizer_unbiased}
\overline{X}&=\exp\Big(-i\sqrt{\pi}(\hat{q}_0-\sum_{k=n+1}^{2n}\hat{p}_k)\Big),\\
\overline{Z}&=\exp\Big(i\sqrt{\pi}(-\hat{p}_0+\sum_{k=1}^{n}\hat{p}_k)\Big).
\end{eqs}
The stabilizers are
\begin{eqs}
\begin{cases}
\hat{S}_{0,q}&=\overline{Z}^2=\exp\Big(-2i\sqrt{\pi}(\hat{q}_0-\sum_{k=n+1}^{2n}\hat{p}_k)\Big) \\
\hat{S}_{0,p}&=\overline{X}^2=\exp\Big(2i\sqrt{\pi}(-\hat{p}_0+\sum_{k=1}^{n}\hat{p}_k)\Big)\\
\hat{S}_{j,q}&=\exp\Big(i\sqrt{2\pi}(\hat{q}_j+\hat{q}_0)\Big), \qquad\qquad\; 1\leq j\leq n\\
\hat{S}_{l,q}&= \exp\Big(i\sqrt{2\pi}(\hat{q}_l+\sum_{k=1}^n\hat{p}_k-\hat{p}_0)\Big), n+1\leq j\leq 2n\\
\hat{S}_{s,p}&= \exp\Big(-i\sqrt{2\pi}\hat{p}_s\Big), \qquad\qquad\qquad 1\leq s \leq 2n
\end{cases}.
\end{eqs}
Here, $\hat{S}_{0,q}, \hat{S}_{0,p}$ are the GKP \textcolor{black}{stabilizers} of logical modes that will be used in the \textcolor{black}{second} layer \textcolor{black}{of} error correction in scheme III. The last three lines of Eq.~\eqref{eq: stabilizer_unbiased} are the canonical GKP stabilizers used in the \textcolor{black}{first} layer \textcolor{black}{of} error correction.
Similar to the canonical GKP-repetition codes, we use maximum-likelihood to estimate the most probable quadrature error $\vec{\xi}_*$ compatible with a given measurement result $(z_{1,q},~...,~z_{2n,q},~z_{1,p},~...,~z_{2n,p})$ such that $|\vec{\xi}|^2$ is minimized. Since the variance of each quadrature noise is proportional to $n$, the small error regime to allow for approximating $R_{\sqrt{2\pi}}(z)=z$ is $\sqrt{n}\sigma \ll \sqrt{2\pi}$.

We can write $|\vec{\xi}|^2$ in terms of $\vec{z}$,
\begin{eqs}
|\vec{\xi}|^2&=(z_{0,p}+\sum_{k=n+1}^{2n}z_{k,p})^2+\sum_{k=1}^{2n}(z_{k,p})^2\\
&+(\sum_{k=1}^n z_{k,p}-z_{0,q})^2+\sum_{k=1}^n (z_{k,q}-z_{0,p}-\sum_{j=n+1}^{2n} z_{j,p})^2\\
&+\sum_{k=n+1}^{2n} (z_{k,q}- z_{0,q})^2.
\end{eqs}
Following the maximum-likelihood error correction of layer one of scheme III (see Sec.~\ref{subsubsec:scheme3-dec}), we obtain the most probable errors on both quadratures
\begin{eqs}
\overline{z}_{0,q}&=\frac{1}{n+1}(\sum_{k=1}^n z_{k,p}+\sum_{k=n+1}^{2n} z_{k,q}),\\
\overline{z}_{0,p}&=\frac{-1}{n+1} \Big[(n+1)\sum_{k=n+1}^{2n} z_{k,p} -\sum_{k=1}^n z_{k,q} \Big],
\end{eqs}
that acts on the logical mode after decoding. The above equation is equivalent to $\begin{pmatrix} \overline{z}_{0,q}\\ \overline{z}_{0,p} \end{pmatrix} =A_2 \vec{\xi}'$.

Finally, we perform error correction by subtracting the actual quadrature error by the estimated values. Finally, the residual noise vector on position and momentum quadratures is
\begin{eqs}
\vec{\xi}_{\text{final},l}= A_2 (\vec{\xi}-\vec{\xi}_*)= \begin{pmatrix}z_{0,q}-\overline{z}_{0,q}\\
z_{0,p}-\overline{z}_{0,p}\end{pmatrix},
\end{eqs}
where
\begin{eqs}
z_{0,q}-\overline{z}_{0,q}&=\frac{1}{n+1}\Big(\xi_{0,p}- \sum_{k=n+1}^{2n} \xi_{k,q} \Big) \sim \mathcal{N}(0,\frac{\sigma^2}{n+1}),\\
z_{0,p}-\overline{z}_{0,p}&=\frac{1}{n+1} (\xi_{0,q}-\sum_{k=1}^n \xi_{k,q}) \sim \mathcal{N}(0,\frac{\sigma^2}{n+1}).
\end{eqs}

This construction can simultaneously suppress logical position and momentum error without extra squeezing. However, unlike the original GKP-repetition code, which has a syndrome with constant variance for a given variance of physical noise $\sigma^2$, the unbiased GKP-repetition
code has a syndrome quadrature whose variance scales as $n$, meaning that at most a linear suppression of logical errors. The construction of a family of canonical GKP-stabilizer codes that can quadratically suppress both logical position and momentum displacement error with syndrome independent of $n$ without squeezing is an interesting open question.

\section{Summary \& discussion}
$$
\begin{array}{c}
\text{\textit{It's not about the code, it's about the decoder.}}\\
\quad\quad\quad\quad\quad\quad\quad\quad\quad\quad\quad\quad\text{-- error-correction lore}
\end{array}
$$
We study and benchmark three schemes encoding qubits into harmonic oscillators utilizing concatenations of various qubit and bosonic encodings. While encoding circuits for all schemes follow a similar pattern, the decoders and error-correcting performance of the schemes are substantially different, resonating with the above quote. 

A key motivation for this work was to gauge the usefulness of the recently discovered GKP-stabilizer codes \cite{Noh_many_oscillator} for encoding discrete-variable information. 
We concatenate these mode-into-mode codes with bona-fide qubit-into-mode GKP codes \cite{gkp} in the third of the three concatenation schemes considered in this work. 
Previous research on the decoding of GKP-stabilizer codes has been conducted on a case-by-case basis. 
In this study, we introduce a formalism for a maximum-likelihood decoder that allows for efficient calculation of the performance of general GKP-stabilizer codes. 
Our approach also serves as a recipe for conducting Monte Carlo simulations, and we provide illustrative examples of this.

We find that the performance of GKP-stabilizer concatenated codes can vary greatly. 

In the case when three concatenation schemes utilize derivatives of the repetition code, GKP-stabilizer codes outperform the two conventional schemes.
The performance of all three schemes does not depend very much on the (classical) distance of the (bit-flip) repetition code since, in all cases, the logical error rate contributions are dominated by phase-flip errors. 

In cases when three concatenation schemes utilize derivatives of other codes, such as the $\llbracket 5,1,3\rrbracket $, $\llbracket 7,1,3\rrbracket $, or $\llbracket 9,1,3\rrbracket $ codes, the GKP-stabilizer scheme comes in either second or third place.
We conjecture that GKP-stabilizer scheme performance is very sensitive to the weight of the stabilizers of the underlying code.
GKP-stabilizer error correction was originally designed for the limit of small Gaussian fluctuations.
Larger-weight stabilizers collect fluctuations coming from more modes, concentrating noise in a way that this limit no longer applies. 
In such cases, we thus expect (and observe) a degradation in the performance of GKP-stabilizer error correction.

In the process of benchmarking our schemes, we recast maximum-likelihood error-correction against quadrature fluctuations to a statistical estimation problem for GKP codes, GKP-stabilizer codes, and analog stabilizer codes. We believe that the resulting statistical inference problem is related to compressed sensing. Since compressed sensing has been studied in the context of quantum tomography \cite{david10compressed}, it may be interesting to investigate the underlying connections between our schemes and that work.

Our schemes all rely on the use of GKP states, which we set to be those associated with a square lattice in a mode's phase space. However, given that our decoding processes distribute the initially uniform displacement noise in a lopsided manner among the modes, there is good reason to consider initializing auxiliary modes to GKP states associated with rectangular lattices. The shape of the lattice can be added as an extra parameter in the statistical estimation problem, potentially yielding a more effective decoder.

Our results \textcolor{black}{are} applicable to the finite-energy GKP states in a realistic setting by introducing another Gaussian noise channel to the initialization stage, because the finite-energy GKP state is usually implemented by applying an envelope operator $\exp(-\Delta \hat{n})$ to an ideal GKP state \cite{cahill1969expansion,noh2020fault,nick14fault}. For $\Delta\ll 1$ regime, we can expand the envelope operator in terms of a Gaussian displacement channel.

Towards experimental realizations, the three major obstacles are 1. GKP state preparation, 2. syndrome measurements of GKP stabilizers, 3. two-mode Gaussian operations. The most promising platforms for realizing GKP error correction are microwave cavities, optical systems, and phononic modes in ion traps.


Multi-mode Gaussian operations have long been realized in the optical domain \cite{Furusawa_Sum2008}, with the primary remaining difficulty being the preparation of GKP states. 
For microwave cavity systems, GKP state preparation and syndrome measurement have recently been implemented  \cite{campagne2020quantum, break_even2022}, and two-cavity Gaussian operations have been demonstrated in other applications of cavity and circuit-QED \cite{pfaff2017controlled,gao2019entanglement,wang2020efficient,grimm2019kerr,roy2016introduction}. 
For ion trap systems, the GKP state preparation and syndrome measurement have also been simultaneously realized \cite{fluhmann18sequential, fluhmann2019encoding, de2022error}.
Recently, two-mode Gaussian operations have been achieved as well \cite{chen2023scalable}.
Overall, with these recent advances, it seems very likely that multimode GKP codes will be implemented in these technologies in the near future.

We have benchmarked our concatenated schemes against displacement noise only, which is just the tip of the iceberg.
Aside from displacement error, loss and dephasing errors are also prevalent in physical systems. 
While GKP and GKP-stabilizer codes have been developed with translational-type displacement noise in mind, it may \textcolor{black}{be} interesting to generalize our schemes to work against dephasing errors using auxiliary modes in rotation-symmetric states such as number-phase states or cat states \cite{grimsmo_rotation}. 
Another direction would be to adapt our schemes to biased noise \cite{tuckett18ultrahigh,tuckett19tailor,tuckett20ft,dua22clifford,xu22tailor}, utilizing squeezing and/or highly deformed GKP lattices \cite{Hanggli_biased,conrad2022gottesman,wu2022optimal}.

\begin{acknowledgements}
We thank Kyungjoo Noh, Mohammad Hafezi, and  Anthony J. Brady for helpful discussions, as well as Henry Pfister for passing on error-correction lore in the form of the quote mentioned at the end of this manuscript. Y.X. and E.-J.K. are supported by ARO W911NF-15-1-0397, National Science Foundation QLCI grant OMA-2120757, AFOSR-MURI FA9550-19-1-0399, Department of Energy QSA program. Y.W. is supported by the Air Force Office of Scientific Research under award number FA9550-19-1-0360.
Y.X. thanks Michael Gullans and Alexander Barg for teaching the classical and quantum error correction courses. Y.X. would like to thank Yilun Li and Yujie Zhang for their mental support during the COVID-19 pandemic. V.V.A. thanks Olga Albert and Ryhor Kandratsenia for providing daycare support throughout this work.
\end{acknowledgements}

\appendix

\section{Analytical expression of logical error rates in 3-qubit repetition code}\label{appendix:analytic-3-rep}

In this appendix, we give analytical expressions for the logical error rates of the 3-qubit repetition code in different schemes. Here we assume the position and the momentum displacement errors on each oscillator obey the same Gaussian distribution $p_{\sigma}(\xi)=(2\pi\sigma^2)^{-1/2}\exp{(-\xi^2/2\sigma^2)}$.

For the $\llbracket 5,1,3\rrbracket $ code, calculations can be performed using the same analysis. They are more complicated in expression as more modes are involved and position momentum errors are mixed, so we will not present the result.

\subsection{Scheme I}
Scheme I is analogous to the usual stabilizer code. We first calculate the error rate of $p$ and $q$ on each physical GKP qubit, then calculate the logical error rate.

For a single physical qubit, if the displacement error $\xi_{p,q_i}\in[(2n-1/2)\sqrt{\pi},(2n+1/2)\sqrt{\pi})$, it can be corrected without introducing an error. As we suppose the position and momentum error obey the same Gaussian distribution, the success probability is
\begin{equation}
    p_{0}=\sum_{n\in\mathbb{Z}}\int^{(2n+1/2)\sqrt{\pi}}_{(2n-1/2)\sqrt{\pi}}d\xi_{p,q_i}p_{\sigma}(\xi_{p,q_i}).
\end{equation}

The 3-qubit repetition code can correct at most one $X$ error. This means only when there is no $q$ error on any physical GKP and at most one $p$ error is there no logical error. So the logical error rate is
\begin{equation}
    \text{Prob}_{\text{1}}(\text{logical error})=1-(p_0^3+3p_0(1-p_0)^2)(p_0^3+3p_0^2(1-p_0)).
\end{equation}

\subsection{Scheme II}
After decoding, the logical information is \textcolor{black}{in} the first oscillator mode. So we can make use of Eq. \eqref{scheme2finalerror} to calculate the final logical error distribution in terms of the errors on the physical qubits.
\begin{eqs}
\xi_{q}&=\frac 13(\xi_{q_1}+\xi_{q_2}+\xi_{q_3})-R_{\sqrt{\pi}}(\frac 13(\xi_{q_1}+\xi_{q_2}+\xi_{q_3}))\\
\xi_{p}&=\xi_{p_1}+\xi_{p_2}+\xi_{p_3}-R_{\sqrt{\pi}}(\xi_{p_1}+\xi_{p_2}+\xi_{p_3}).
\end{eqs}
It is useful to make the variable change $x_{p/q}=\frac 13(\xi_{p_1/q_1}+\xi_{p_2/q_2}+\xi_{p_3/q_3})$, $y_{p/q}=\xi_{p_2/q_2}-\xi_{p_1/q_1}$, $z_{p/q}=\xi_{p_3/q_3}-\xi_{p_1/q_1}$. The probability density function of logical quadrature noise $\xi_q$ and $\xi_p$ are 
\begin{widetext}
\begin{eqs}
Q_2(\xi_q)&=\int^{\infty}_{-\infty} d\xi_{q_1}\int^{\infty}_{-\infty} d\xi_{q_2}\int^{\infty}_{-\infty} d\xi_{q_3}p_{\sigma}(\xi_{q_1})p_{\sigma}(\xi_{q_2})p_{\sigma}(\xi_{q_3})\delta(\xi_{q}-\frac 13(\xi_{q_1}+\xi_{q_2}+\xi_{q_3})+R_{\sqrt{\pi}}(\frac 13(\xi_{q_1}+\xi_{q_2}+\xi_{q_3})))\\
&=\int^{\infty}_{-\infty} d x_q\int^{\infty}_{-\infty} dy_q\int^{\infty}_{-\infty} dz_q~ p_{\sigma}(x_q-\frac {y_q}{3}-\frac {z_q}{3})p_{\sigma}(x_q+\frac {2y_q}{3}-\frac {z_q}{3})p_{\sigma}(x_q-\frac {y_q}3+\frac {2z_q}{3})\delta(\xi_{q}-x_q+R_{\sqrt{\pi}}(x_q))\\
&=\sum_{n\in\mathbb{Z}}\int^{(n+\frac 12)\sqrt{\pi}}_{(n-\frac 12)\sqrt{\pi}} d x_q\int^{\infty}_{-\infty} dy_q\int^{\infty}_{-\infty} dz_q~ p_{\sigma}(x_q-\frac {y_q}{3}-\frac {z_q}{3})p_{\sigma}(x_q+\frac {2y_q}{3}-\frac {z_q}{3})p_{\sigma}(x_q-\frac {y_q}3+\frac {2z_q}{3})\delta(\xi_{q}+n\sqrt{\pi}),\\
P_2(\xi_p)&=\sum_{n\in\mathbb{Z}}\int^{\frac 13(n+\frac 12)\sqrt{\pi}}_{\frac 13(n-\frac 12)\sqrt{\pi}} d x_p\int^{\infty}_{-\infty} dy_p\int^{\infty}_{-\infty} dz_p~ p_{\sigma}(x_p-\frac {y_p}{3}-\frac {z_p}{3})p_{\sigma}(x_p+\frac {2y_p}{3}-\frac {z_p}{3})p_{\sigma}(x_p-\frac {y_p}3+\frac {2z_p}{3})\delta(\xi_{p}+n\sqrt{\pi}).
\end{eqs}
This means after error correction and decoding, the logical errors taks discrete values of $n\sqrt{\pi},~ n\in\mathbb{Z}$. When $n$ takes even numbers, there is no error. so 
\begin{eqs}
\text{Prob}_{\text{2}}(q ~\text{is correct})=&\sum_{n\in\mathbb{Z}}\int^{(2n+\frac 12)\sqrt{\pi}}_{(2n-\frac 12)\sqrt{\pi}} d x_q\int^{\infty}_{-\infty} dy_q\int^{\infty}_{-\infty} dz_q~ p_{\sigma}(x_q-\frac {y_q}{3}-\frac {z_q}{3})p_{\sigma}(x_q+\frac {2y_q}{3}-\frac {z_q}{3})p_{\sigma}(x_q-\frac {y_q}3+\frac {2z_q}{3}),\\
\text{Prob}_{\text{2}}(p~\text{is correct})=&\sum_{n\in\mathbb{Z}}\int^{\frac 13(2n+\frac 12)\sqrt{\pi}}_{\frac 13(2n-\frac 12)\sqrt{\pi}} d x_p\int^{\infty}_{-\infty} dy_p\int^{\infty}_{-\infty} dz_p~ p_{\sigma}(x_p-\frac {y_p}{3}-\frac {z_p}{3})p_{\sigma}(x_p+\frac {2y_p}{3}-\frac {z_p}{3})p_{\sigma}(x_p-\frac {y_p}3+\frac {2z_p}{3}).
\end{eqs}

The final logical error rate is
\begin{equation}
   \text{Prob}_{\text{2}}(\text{logical error})=1-\text{Prob}_{\text{2}}(p~\text{is correct})\text{Prob}_{\text{2}}(q~\text{is  correct}).
\end{equation}

\subsection{Scheme III}
Similar to scheme II, we make use of Eq. \eqref{eq:finalerror3} to write the logical error as
\begin{eqs}
\xi_{q}&=\xi_{q_1}+\frac 13 (R_{\sqrt{2\pi}}(\xi_{q_2}-\xi_{q_1})+R_{\sqrt{2\pi}}(\xi_{q_3}-\xi_{q_1})),\\
\xi_{p}&=\xi_{p_1}+\xi_{p_2}+\xi_{p_3}-R_{\sqrt{2\pi}}(\xi_{p_2})-R_{\sqrt{2\pi}}(\xi_{p_3}).
\end{eqs}

The probability density function of logical quadrature noise $\xi_q$ and $\xi_p$ are

\begin{eqs}
Q_3(\xi_q)=&\int^{\infty}_{-\infty} d\xi_{q_1}\int^{\infty}_{-\infty} d\xi_{q_2}\int^{\infty}_{-\infty} d\xi_{q_3}p_{\sigma}(\xi_{q_1})p_{\sigma}(\xi_{q_2})p_{\sigma}(\xi_{q_3})\delta(\xi_{q}-\xi_{q_1}-\frac 13 (R_{\sqrt{2\pi}}(\xi_{q_2}-\xi_{q_1})+R_{\sqrt{2\pi}}(\xi_{q_3}-\xi_{q_1})))\\
=&\int^{\infty}_{-\infty} d x_q\int^{\infty}_{-\infty} dy_q\int^{\infty}_{-\infty} dz_q~ p_{\sigma}(x_q-\frac {y_q}{3}-\frac {z_q}{3})p_{\sigma}(x_q+\frac {2y_q}{3}-\frac {z_q}{3})p_{\sigma}(x_q-\frac {y_q}3+\frac {2z_q}{3})\\&\delta(\xi_{q}-x_q-\frac 13 (y_q-R_{\sqrt{2\pi}}(y_q)+z_q-R_{2\sqrt{\pi}}(z_q)))\\
=&\sum_{n_y\in\mathbb{Z}}\sum_{n_z\in\mathbb{Z}}\int^{(n_y+\frac 12)\sqrt{2\pi}}_{(n_y-\frac 12)\sqrt{2\pi}} d y_q\int^{(n_z+\frac 12)\sqrt{2\pi}}_{(n_z-\frac 12)\sqrt{2\pi}} d z_q~ p_{\sigma}(\xi_q-\frac{\sqrt{2\pi}}{3}(n_y+n_z)-\frac {y_q}{3}-\frac {z_q}{3})p_{\sigma}(\xi_q-\frac{\sqrt{2\pi}}{3}(n_y+n_z)+\frac {2y_q}{3}-\frac {z_q}{3})\\&p_{\sigma}(\xi_q-\frac{\sqrt{2\pi}}{3}(n_y+n_z)-\frac {y_q}3+\frac {2z_q}{3}),\\
P_3(\xi_p)=&\sum_{n_y\in\mathbb{Z}}\sum_{n_z\in\mathbb{Z}}\int^{(n_y+\frac 12)\sqrt{2\pi}}_{(n_y-\frac 12)\sqrt{2\pi}} d \xi_{p_2}\int^{(n_z+\frac 12)\sqrt{2\pi}}_{(n_z-\frac 12)\sqrt{2\pi}} d \xi_{p_3}~ p_{\sigma}(\xi_{p_2})p_{\sigma}(\xi_{p_3})p_{\sigma}(\xi_p-\sqrt{2\pi}(n_y+n_z)).
\end{eqs}
\end{widetext}
With the logical error distribution, one can calculate 
\begin{eqs}
&\text{Prob}_{\text{3}}(q~\text{is correct})=\sum_{n\in\mathbb{Z}}\int^{(2n+\frac 12)\sqrt{\pi}}_{(2n-\frac 12)\sqrt{\pi}}d\xi_q~Q_3(\xi_{q}),\\
&\text{Prob}_{\text{3}}(p~\text{is correct})=\sum_{n\in\mathbb{Z}}\int^{(2n+\frac 12)\sqrt{\pi}}_{(2n-\frac 12)\sqrt{\pi}}d\xi_p~P_3(\xi_{p}),\\
&\text{Prob}_{\text{3}}(\text{logical error})=1-\text{Prob}_{\text{3}}(p~\text{is correct})\text{Prob}_{\text{3}}(q~\text{is correct}).
\end{eqs}

\section{Analysis on the logical GKP stabilizer error correction}\label{appendix:logicalGKPEC}

In this section we show that if the logical GKP stabilizer error correction is the last layer, it will not change the final error rate. 

We first establish the following lemma of integral equivalence.

\begin{lemma}[Equivalence of integrals]\label{lemma:equivalence}
Let $f:V_{2n}\to V_{2n}$ be a function such that $\vec{\eta}=\vec{f}(\vec{\xi})$. Let $\eta_l$ be its component and we write $\eta_l=f_l(\vec{\xi})$. Let $d^{2n}\xi\equiv d\xi_{1}...d\xi_{2n}$ be the integration measure and $\mathfrak{p}(\vec{\xi})$ be the probability distribution of
the vector $\vec{\xi}$. Then $ \forall n\in\mathbb{Z}$, the following two integrals are equivalent:
\begin{eqs}\label{eq:lemmaintegralequiv}
&\int_{(n-1/2)\sqrt{\pi}}^{(n+1/2)\sqrt{\pi}}d\eta_l \int d^{2n}\xi\, p(\vec{\xi})\delta(\eta_l-f_l(\vec{\xi}))\\
=&\int_{(n-1/2)\sqrt{\pi}}^{(n+1/2)\sqrt{\pi}}d\eta_l \int d^{2n}\xi\, p(\vec{\xi})\delta(\eta_l-f_l(\vec{\xi})+R_{\sqrt{\pi}}(f_l(\vec{\xi}))).
\end{eqs}
\end{lemma}

\begin{proof}
Integrating against $d\xi_l$, the claimed equivalence in Eq. \eqref{eq:lemmaintegralequiv} is transformed to 
\begin{equation}
   \int_{V_1} d^{2n}\xi\, p(\vec{\xi})=\int_{V_2} d^{2n}\xi\, p(\vec{\xi}), 
\end{equation}
where $V_1=\{\vec{\xi}|f_l(\vec{\xi})\in[(n-1/2)\sqrt{\pi},(n+1/2)\sqrt{\pi})\}$, and $V_2=\{\vec{\xi}|f_l(\vec{\xi})-R_{\sqrt{\pi}}(f_l(\vec{\xi}))\in[(n-1/2)\sqrt{\pi},(n+1/2)\sqrt{\pi})\}$. It is sufficient to show that $V_1=V_2$. $V_1\subseteq V_2$ is obvious. On the other hand, if $f_l(\vec{\xi})-R_{\sqrt{\pi}}(f_l(\vec{\xi}))\in[(n-1/2)\sqrt{\pi},(n+1/2)\sqrt{\pi})$, since $f_l(\vec{\xi})-R_{\sqrt{\pi}}(f_l(\vec{\xi}))$ is an integer multiple of $\sqrt{\pi}$, we must have $f_l(\vec{\xi})-R_{\sqrt{\pi}}(f_l(\vec{\xi}))=n\sqrt{\pi}$. By definition $R_{\sqrt{\pi}}(f_l(\vec{\xi}))\in[-1/2\sqrt{\pi},1/2\sqrt{\pi})$, we have $f_l(\vec{\xi})=n\sqrt{\pi}+R_{\sqrt{\pi}}(f_l(\vec{\xi}))\in[(n-1/2)\sqrt{\pi},(n+1/2)\sqrt{\pi})$. So $V_2\subseteq V_1$, and $V_1=V_2$. Hence we proved Eq.\eqref{eq:lemmaintegralequiv}.
\end{proof}

We now investigate the cases of scheme II and \textcolor{black}{III}. In scheme II, if we only do the first round error correction, the final error distribution will be $\vec{\xi}_{\text{final}}=A_{\text{enc}}P_{A_1}^{\perp}\vec{\xi}$, c.f. Eq.~\eqref{eq:post-EC-error}. If we only focus on the position and momentum error of the logical modes, we subtract the corresponding rows and get $\vec{\xi}^{(1)}_{\text{final},l}=A_{2}P_{A_1}^{\perp}\vec{\xi}$. In scheme II, the final error after two rounds of error corrections is Eq.~\eqref{scheme2finalerror}. Again, focusing on the errors of logical modes, we have 
\begin{eqs}
   \vec{\xi}^{(2)}_{\text{final},l}&= A_{2}\left(P_{A_1}^{\perp}\vec{\xi}-(A_2P_{A_1}^{\perp})^T\left(A_2P_{A_1}^{\perp}A_2^T\right)^{-1}R_{\sqrt{\pi}}(A_2P_{A_1}^{\perp}\vec{\xi})\right)\\
   &=A_{2}P_{A_1}^{\perp}\vec{\xi}-R_{\sqrt{\pi}}(A_2P_{A_1}^{\perp}\vec{\xi}).
\end{eqs}

Comparing $\vec{\xi}^{(1)}_{\text{final},l},~\vec{\xi}^{(2)}_{\text{final},l}$ with Eq.~\eqref{eq:lemmaintegralequiv}, we see that $\vec{f}(\vec{\xi})=A_{2}P_{A_1}^{\perp}\vec{\xi}$. Recall that for the position or momentum error of $j$th logical mode, its correct rate can be written as
\begin{equation}
    \sum_n \int_{(2n-1/2)\sqrt{\pi}}^{(2n+1/2)\sqrt{\pi}}d\xi_{j,p/q} \int d^{2n}\xi p(\vec{\xi})\delta(\xi_{j,p/q}-\left(\vec{\xi}^{(1/2)}_{\text{final},l}\right)_{j,p/q}).
\end{equation}
So summing over all the even integers in Eq.~\eqref{eq:lemmaintegralequiv}  using either $\vec{\xi}^{(1)}_{\text{final},l}$ or $\vec{\xi}^{(2)}_{\text{final},l}$ produces the same result. This means that the correct rate of each logical mode does not change after the second round of logical GKP stabilizer error correction.

For scheme III, as discussed in Section \ref{subsubsec:scheme3-dec}, if we stop after only the first layer, the errors on the logical modes are Eq.~\eqref{eq:post-EC-error-III}
\begin{equation}
  \vec{\xi}^{(1)}_{\text{final},l}=A_2(\vec{\xi}-\vec{\xi}_*)=A_2\vec{\xi}-A_2A_3^T(A_3A_3^T)^{-1}R_{\sqrt{2\pi}}(A_{3}\vec{\xi}).
\end{equation}

After applying Layer 2, the final errors on the logical modes are in Eq.~\eqref{eq:finalerror3}. We repeat it below:
\begin{eqs}
&\vec{\xi}^{(2)}_{\text{final},l}=A_2(\vec{\xi}-\vec{\xi}_{*}^{\prime})\\
=&A_2\vec{\xi}-A_2A_3^T(A_3A_3^T)^{-1}\vec{z}-R_{\sqrt{\pi}}(A_{2}(\vec{\xi}-A_3^T(A_3A_3^T)^{-1}\vec{z})),
\end{eqs}
where $\vec{z}=R_{\sqrt{2\pi}}(A_3\vec{\xi})$. Now again comparing $\vec{\xi}^{(1)}_{\text{final},l},~\vec{\xi}^{(2)}_{\text{final},l}$ and Eq.\eqref{eq:lemmaintegralequiv}, it is easy to see that in the case of scheme III, $\vec{f}(\vec{\xi})=A_2\vec{\xi}-A_2A_3^T(A_3A_3^T)^{-1}R_{\sqrt{2\pi}}(A_3\vec{\xi})$. Applying the same argument as in scheme II , we see that for scheme III the second layer of logical GKP error correction will not change the correct rate of each logical mode either.

Though the discussion above focuses on the position or momentum error of a single mode, it is not difficult to see that the error rate on a multi-mode code subspace should be unchanged under the setting of this appendix.  

\section{Numerical Simulations}\label{appendix:numerics}

\subsection{Methods}

In this section, we discuss how to numerically simulate different error correction schemes using the Monte Carlo method. The method consists of three steps and is repeated \(M\) times.

\textbf{Initialization}: 
First, we initialize two vectors to store the displacement noise vector Eqs.~\eqref{eq:displacement_noise},~\eqref{eq:noise_vector} acting on codeword qubits, i.e., the vector  \((\vec{\xi}_{q}^{n}|\vec{\xi}_{p}^{n})\) consisting of two \(n\)-dimensional vectors $\vec{\xi}_q^n$ and $\vec{\xi}_p^n$. 
The displacement error acting on the codeword qubits is characterized by the \(2n\)-dimensional covariance matrix $\sigma$.

\textbf{Error correction}: The error correction procedure uses the update rules we discussed in Sec.~\ref{sec:EC_dec}, resulting in the residual noise vector $\vec{\xi}_{\text{final}}$.

After error correction, we decode the error-corrected quadrature and obtain the final residual noise vector $\vec{\xi}_{\text{final}}$ \eqref{eq:finalerror}.  
If $ \exists i\in \{1,2,...,k\}$, such that $R_{2\sqrt{\pi}}(\xi_{j,q})\geq \frac{\sqrt{\pi}}{2}$ or $R_{2\sqrt{\pi}}(\xi_{j,p})\geq \frac{\sqrt{\pi}}{2}$, then there is an error, and we update the error rate accordingly:
\begin{eqs}
\text{logical error rate}\leftarrow \text{logical error rate}+\frac{1}{M}~.
\end{eqs}

After completing the above $M$ times, we call the resulting value the empirical logical error rate $p_{\text{emp}}$.

For the error of Monte Carlo simulations, we regard the sampling process as a binomial distribution: the actual logical error rate is $p$ for each sample, where $p$ depends on the decoding schemes. The binomial distribution is
\begin{eqs}
\text{Prob}(\text{no logical error})&=\text{Prob}(x=0)=1-p\\
\text{Prob}(\text{logical error})&=\text{Prob}(x=1)=p,
\end{eqs}
where $x$ is the frequency of logical error. We define the empirical logical rate $p_{\text{emp}}$ as
\begin{eqs}
\frac{\sum_{i=1}^M x_i}{M}=\frac{N_{\text{error}}}{M}=p_{\text{emp}},
\end{eqs}
where $N_{\text{error}}$ is the number of samples with logical errors.

If we take $M$ independent and identical samples, the probability of getting an empirical logical error rate $p_{\text{emp}}$ is

\begin{eqs}
\binom{N_{\text{error}}}{M} p^{N_{\text{error}}} (1-p)^{M-N_{\text{error}}}.
\end{eqs}
The central limit theorem tells us that \textcolor{black}{the} mean value of $x$ will converge to $p$ and the variance of $x/M$ will converge to $p(1-p)/M$ when $M\rightarrow \infty$. Hence, we plot error bars of size $\pm \sqrt{p_{\text{emp}}(1-p_{\text{emp}})/M}$ for the numerical studies in the next section.
    
\subsection{More numerical results}\label{appendix:more_numerics}

Fig.~\ref{fig:multi_repetition} shows the numerical simulation of the scheme I and III for 5-qubit and 7-qubit repetition codes. Due to the limit of numerical simulation, we do not include the simulation of scheme III for $\sigma<0.18$. In these codes, under the assumption of our unbiased noise model, the final logical error is \textcolor{black}{mostly due to} the phase error of one bit in both schemes. So we expect the crossover of the logical error rates between two schemes will happen at roughly the same $\sigma$, regardless of the number of physical modes.

\begin{figure}[h]
    \centering
    \includegraphics[width=0.4\textwidth]{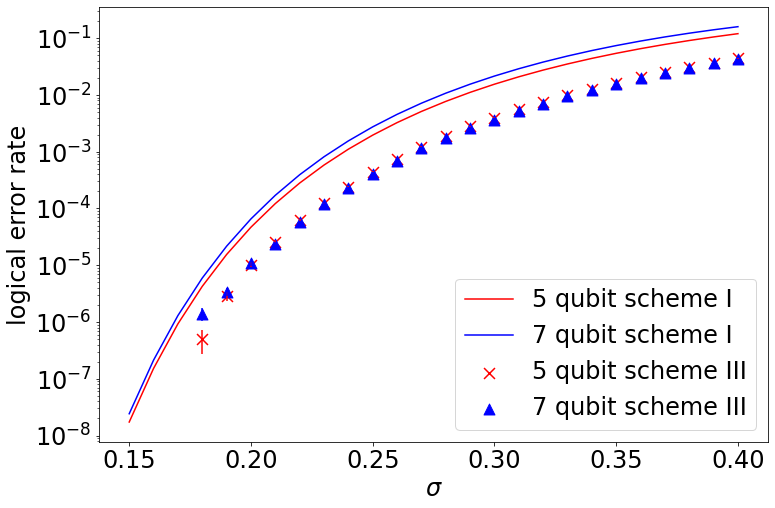}
    \caption{Comparison between scheme I and III for 5-qubit and 7-qubit repetition codes. The result of scheme I is calculated analytically, and the result of scheme III is obtained by Monte Carlo simulation.}
    \label{fig:multi_repetition}
\end{figure}

Fig.~\ref{fig:squeezed_scheme13} demonstrates the changing of the logical error rates with respect to the shape of the GKP lattice. The parameter $a$ is defined so the squeezing changes $\hat{p}\rightarrow \sqrt{\alpha} \hat{p}$, $\hat{q}\rightarrow \frac{1}{\sqrt{\alpha}} \hat{q}$. Correspondingly, the period in $\hat{q}$ is multiplied by \textcolor{black}{$\sqrt{\alpha}$} while the period in $\hat{p}$ is multiplied by \textcolor{black}{$1/\sqrt{\alpha}$}. This shows that the \textcolor{black}{distinguishability} of syndrome measurements can \textcolor{black}{be} improved by applying squeezing to GKP modes.

\begin{figure}[h]
    \centering
    \includegraphics[width=0.45\textwidth]{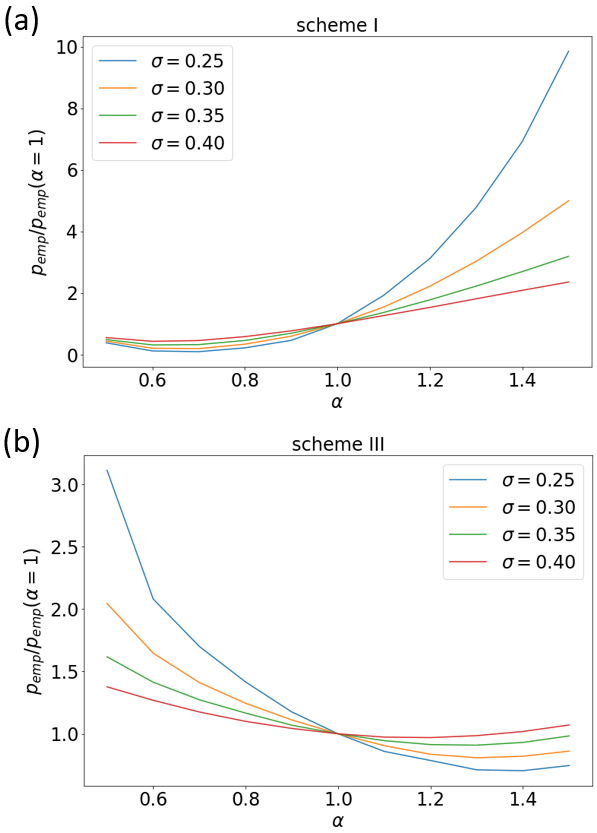}
    \caption{The performance of schemes I and III for rectangle GKP lattices. In part (a), we simultaneously squeeze all the modes $\hat{p}\rightarrow \sqrt{\alpha} \hat{p}$, $\hat{q}\rightarrow \frac{1}{\sqrt{\alpha}} \hat{q}$ for scheme I. In part (b), we only squeeze the auxiliary mode $\hat{p}\rightarrow \sqrt{\alpha} \hat{p}$, $\hat{q}\rightarrow \frac{1}{\sqrt{\alpha}} \hat{q}$. The vertical axis in both subplots reflects the ratio of the logical error rate to the logical error rate of the unsqueezed ($\alpha=1$) scenario.}
    \label{fig:squeezed_scheme13}
\end{figure}

\section{Qudit version of GKP-stabilizer codes}\label{appendix:qudit_version}

In this section, we demonstrate the qudit version of the GKP-stabilizer code. Here, we
define the qudit Pauli matrices to be

\begin{eqs}
X=\sum_{j=0}^{d-1} \ket{j+1}\bra{j},~ Z=\sum_{j=0}^{d-1} \omega^{j} \ket{j}\bra{j},
\end{eqs}
where $\omega=e^{i2\pi/d}$.
We define the CNOT gate to be

\begin{eqs}
\text{CNOT}_{1\rightarrow 2} \ket{x}\
\otimes \ket{y}= \ket{x}\otimes\ket{(x+y)\mod d}.
\end{eqs}

Here we use the two-mode GKP-repetition code \cite{Noh_many_oscillator} that encodes the logical information of a data qudit into a two-qudit system as an example. Let the quantum state of the data qubit to be $\ket{\psi}=\sum_{x=0}^{d-1} \ket{x}$ where $\ket{x}$ is the eigenstate of $Z$ with eigenvalue $\omega^{x}$.

The auxiliary qudit is initialized to a canonical qudit-GKP state

\begin{eqs}
\ket{\text{GKP}}_{\text{qudit}}=\frac{1}{\sqrt{d/r}}\sum_{m=0}^{d/r-1} \ket{rm}
\end{eqs}
which is stabilized by $X^r$ and $Z^r$. The second stabilizer implies an addition condition that $\frac{2\pi r^2}{d} \mod 2\pi =0$.

Similar to the regular two-mode canonical GKP-repetition code, the encoding circuit is a $\text{CNOT}_{1\rightarrow 2}$

\begin{eqs}
\text{CNOT}_{1\rightarrow 2} \ket{\psi}\otimes \ket{\text{GKP}}_{\text{qudit}}=\sum_{x=0}^{d-1} \sum_{m=0}^{d/r-1} \frac{\psi(x)}{\sqrt{d/r}}\ket{x}\ket{rm+x}.
\end{eqs}

Then we apply additive Pauli error $X_1^{a_1} Z_1^{c_1} X_2^{a_2} Z_2^{c_2}$. $a_1, c_1, a_2, c_2$ are independent and identical zero-mean random variables.

\begin{eqs}
&\ket{\Phi}=X_1^{a_1} Z_1^{c_1} X_2^{a_2} Z_2^{c_2}\sum_{x=0}^{d-1} \sum_{m=0}^{d/r-1} \frac{\psi(x)}{\sqrt{d/r}}\ket{x}\ket{rm+x}\\
=&\frac{1}{\sqrt{d/r}}\sum_{x=0}^{d-1} \sum_{m=0}^{d/r-1} e^{i\omega [(rm+x) c_2 + x c_1]}\psi(x)\ket{x+a_1} \ket{rm+x+a_2}
\end{eqs}

Then we apply the decoding circuit

\begin{eqs}
&\text{CNOT}_{1\rightarrow 2}^\dagger \ket{\Phi}\\
=& \Big(\sum_{x=0}^{d-1} e^{i \omega x (c_1 + c_2)} \psi(x)\ket{x+a_1} \Big)  \Big(\sum_{m=0}^{d/r-1} \frac{e^{i\omega rm c_2}}{\sqrt{d/r}} \ket{rm+a_2-a_1}\Big)\\
=&\Big(X_1^{a_1} Z_1^{c_1+c_2} \ket{\psi} \Big) \Big(X_2^{a_2-a_1} Z_2^{c_2} \ket{\text{GKP}}_{\text{qudit}}\Big).
\end{eqs}

Since the code distance of ancilla is $r$, if $|a_2-a_1|$ and $|c_2|$ are smaller than $r/2$, then we can extract the $a_2-a_1$ and $c_2$ by measuring the stabilizer of auxiliary qudit ($X_2^r$ and $Z_2^r$). Hence we can correct the $X$-error acting on data qudit by applying error correction $Z_1^{-c_2} X_1^{\frac{1}{2}(a_2-a_1)}$. By assuming $a_2-a_1$ and $c_2$ lie in the unambiguously distinguishable range $[-r/2, r/2]$, this error-correcting code can reduce the variance of $X$ error acting on data qudit by $50\% $ without amplifying the variance of $Z$-error.

\bibliography{biblo.bib}
\end{document}